\documentclass[journal,final,twocolumn,twoside]{IEEEtran}
\usepackage{amsmath,amsfonts}
\usepackage{algorithmic}
\usepackage{algorithm}
\usepackage{amssymb}
\usepackage{subfigure}

\allowdisplaybreaks[4]
\usepackage{float} 

\usepackage{color}
\usepackage{array}
\usepackage[caption=false,font=normalsize,labelfont=sf,textfont=sf]{subfig}
\usepackage{textcomp}
\usepackage{stfloats}
\usepackage{url}
\usepackage{bm}
\usepackage{verbatim}
\usepackage{multirow} 
\usepackage{ntheorem}
\usepackage{graphicx}
\usepackage{cite}
\theorembodyfont{\upshape}
\theoremheaderfont{\it}
\qedsymbol{\square}
\theoremseparator{:}
\newtheorem{theorem}{\indent Theorem}
\newtheorem{lemma}{\indent Lemma}
\newtheorem*{proof}{\indent Proof}
\newtheorem{remark}{\indent Remark}

\newtheorem{proposition}{\indent Proposition}

\makeatletter

\newcommand{\Rmnum}[1]{\expandafter\@slowromancap\romannumeral #1@}
\makeatother
\begin{document}

\title{Combating Interference for Over-the-Air Federated Learning: A Statistical Approach via RIS}

\author{Wei~Shi, Jiacheng~Yao, Wei~Xu,~\IEEEmembership{Fellow,~IEEE},
         Jindan~Xu,~\IEEEmembership{Member,~IEEE}, \\
         Xiaohu~You,~\IEEEmembership{Fellow,~IEEE}, 
         Yonina~C.~Eldar,~\IEEEmembership{Fellow,~IEEE},
         and Chunming~Zhao,~\IEEEmembership{Member,~IEEE}
\thanks{This work was supported in part by the  National Key Research and Development Program under Grant 2020YFB1806608, and in part by the Special Fund for Key Basic Research in Jiangsu Province No. BK20243015. (Corresponding author: Wei Xu).}
\thanks{W. Shi, J. Yao, W. Xu, X. You, and C. Zhao are with the National Mobile Communications Research Laboratory, Southeast University, Nanjing 210096, China, and are also with the Purple Mountain Laboratories, Nanjing 211111, China (e-mail: \{wshi, jcyao, wxu, xhyu, cmzhao\}@seu.edu.cn).}
\thanks{J. Xu is with the School of Electrical and Electronics Engineering, Nanyang Technological University, Singapore 639798, Singapore (e-mail: jindan.xu@ntu.edu.sg).}
\thanks{Y. C. Eldar is with the Faculty of Mathematics and Computer Science, Weizmann Institute of Science, Rehovot 7610001, Israel (e-mail:
yonina.eldar@weizmann.ac.il).}
}

\maketitle

\begin{abstract}
Over-the-air computation (AirComp) integrates analog communication with task-oriented computation, serving as a key enabling technique for communication-efficient federated learning (FL) over wireless networks. However, owing~to its analog characteristics, AirComp-enabled FL (AirFL) is vulnerable to both unintentional and intentional interference. In this~paper, we aim to attain robustness in AirComp aggregation against interference via reconfigurable intelligent surface (RIS) technology to artificially reconstruct wireless environments. Concretely, we establish performance objectives tailored for interference suppression in wireless FL systems, aiming to achieve unbiased gradient estimation and reduce its mean square error (MSE). Oriented at these objectives, we introduce the concept of phase-manipulated favorable propagation and channel hardening for AirFL, which relies on the adjustment of RIS phase shifts to realize statistical interference elimination and reduce the error variance of gradient estimation. Building upon this concept, we propose two robust aggregation schemes of power control and RIS phase shifts design, both ensuring unbiased gradient estimation in the presence of interference. Theoretical analysis of the MSE and FL convergence affirms the anti-interference capability of the proposed schemes. It is observed that computation and interference errors diminish by an order of $\mathcal{O}\left(\frac{1}{N}\right)$ where $N$ is the number of RIS elements, and the ideal convergence rate without interference can be asymptotically achieved by increasing $N$. Numerical results confirm the analytical results and validate the superior performance of the proposed schemes over existing baselines.
\end{abstract}

% \vspace{-10.pt}
\section{Introduction}
\IEEEPARstart{F}{ederated} learning (FL) has been recognized as a promising distributed learning technique to realize ubiquitous edge intelligence for the sixth-generation (6G) wireless networks \cite{2,3}. In a wireless FL system, multiple distributed edge devices are coordinated by a central parameter server (PS) to collaboratively train a global learning model without revealing their local data. More specifically, model parameters are exchanged among edge devices or with the PS, rather than raw data, which reduces the amount of transmitted data and helps protect data privacy. Due to these advantages, the implementation of FL algorithms over wireless networks has been recently studied to support a broad range of intelligent applications \cite{YJC,4,5}. However, the performance of wireless FL is constrained by limited spectrum resource and dynamics of wireless channels \cite{wireless1,wireless2,Eldar1,6}. Especially during uplink model parameters uploading process, communication overhead and latency increase proportionally to the number of participating devices, resulting in a significant performance bottleneck.

To support simultaneous massive uplink transmission and enhance the communication efficiency of wireless FL, over-the-air computation (AirComp) has emerged as a promising solution \cite{9,900,Eldar2}. AirComp merges the concurrent transmission of local updates and model aggregation over the air by exploiting the waveform superposition property of multiple access channels \cite{8,7,pAirFL}. This wireless channel reuse in over-the-air aggregation significantly reduces communication latency and enhances bandwidth utilization, enabling fast-convergent and communication-efficient wireless FL. Recently, several studies have been conducted on the AirComp-enabled FL (AirFL), including power control \cite{11}, device scheduling \cite{12}, and transceiver design\cite{15}.

Although AirComp provides significant performance gains, the analog aggregation nature makes it vulnerable to unintentional/intentional interference. The presence of interference imposes limitations on computational accuracy and consequently impedes the training process. Integrated communication and computation in AirFL contend with significant unintentional interference, including multi-cell interference, full-duplex interference, and multi-task interference. To address multi-cell interference, the authors in \cite{interference_JSAC} quantified FL convergence in the presence of distorted AirComp. Subsequently, cooperative multi-cell optimization is conducted leveraging the analytical findings in order to alleviate interference and balance resources among various FL tasks. Multiple-input multiple-output (MIMO)-based transceiver beamforming is exploited in \cite{multitaskFL} for FL task-oriented interference suppression. In addressing diverse unintentional interference, current solutions often hinge on incorporating a substantial number of antennas to reduce interference via highly directional beams. However, deploying extensive antennas at the transmitter is expensive.

Intentional interference, commonly referred to as malicious attacks, also poses a significant security challenge in AirFL. To cope with malicious attacks in FL, robust aggregation rules have been developed \cite{19,20}, most of which are based on the idea of comparing local updates from different devices and sorting out outliers at the server. However, individual values of local gradients are typically unavailable in AirFL due to the analog superposition of all local gradients over the air. In \cite{22}, a best effort voting (BEV) power control policy was proposed for AirFL by allowing local workers to transmit their gradients with maximum power. It focused on maximizing the transmit power rather than directly suppressing the attacks, which limits performance. The authors of \cite{24} and \cite{25} developed an AirFL transmission framework resilient to Byzantine attacks by introducing the idea of grouping. In this framework, distributed devices are categorized into multiple groups using different wireless resources (i.e., time and frequency) to transmit their model updates, which sacrifices the benefits of AirComp in fully utilizing wireless resources. Therefore, it is necessary to develop an effective robust AirFL framework to suppress the interference while retaining the benefits of AirComp.

As a cost-effective physical-layer technology, reconfigurable intelligent surface (RIS)-aided communications have received extensive investigation \cite{27,28,29,30,jxu,zyhe1}. An RIS comprises a large number of low-cost passive reflecting elements capable of independently controlling the amplitude and phase shifts of incident signals, enabling accurate beamforming \cite{28}. Utilizing its capability to reconstruct wireless propagation environments, RIS can enhance useful signals and suppress interference \cite{29}. Given its potential, RIS-empowered model aggregation for FL has gained significant attention in recent years \cite{RIS-FL-1,RIS-FL-2,RIS-FL-3,RIS-FL-4}. For instance, the authors of \cite{RIS-FL-1} proposed a simultaneous access scheme enabled by RIS to improve model aggregation performance, leading to a communication-efficient FL framework. Although some studies have investigated RIS-assisted AirFL, they mainly focused on link enhancement and beamforming \cite{RIS-AirFL-1,RIS-AirFL-2,RIS-AirFL-3,RIS-AirFL-4,RIS-AirFL-5,RIS-AirFL-6}, neglecting interference suppression. In this paper, we discover the ability of the RIS in terms of statistical interference elimination in AirFL, aimed at enhancing the aggregation robustness against interference. The main contributions of our work are summarized as follows:

\begin{itemize}
\item We establish performance objectives tailored for interference suppression in wireless FL systems to mitigate the impact of interference on FL convergence. Specifically, our objectives are to achieve unbiased gradient estimation while reducing its mean square error (MSE). Meeting these objectives enables rigorous theoretical convergence analysis, which makes it possible for the FL algorithm to achieve rapid convergence to the optimal point.

\item To realize unbiasedness and reduce the MSE of gradient estimation, a new concept of phase-manipulated favorable propagation and channel hardening enabled by RIS is first developed for AirFL. It achieves statistical interference elimination without requiring any instantaneous channel state information at the transmitter (CSIT). Based on this concept, we propose two representative robust aggregation schemes with different power control and RIS phase shift settings for AirFL. Both schemes are shown to be effective in achieving unbiased gradient estimation.

\item Accurate closed-form expressions are derived to evaluate the MSE of gradient estimation for both schemes. The obtained results reveal that increasing the number of RIS reflecting elements, $N$, effectively mitigates the impact of computation, interference, and noise errors by at least an order of $\mathcal{O}\!\left(\frac{1}{N}\right)$. In addition, Scheme~\Rmnum{1} achieves more precise gradient computation, while Scheme~\Rmnum{2} exhibits better efforts in interference and noise suppression.

\item Building upon the derived MSE, the FL convergence of the proposed schemes is analyzed, which confirms a convergence rate on the order of $\mathcal{O}\left(\frac{2\varpi_u}{\sqrt{T}}\right)$, where $T$ is the number of iterations and $\varpi_u$ is a constant related to the specific method. 
% Owing to smaller computation errors in MSE, Scheme \Rmnum{1} always exhibits a faster convergence. Moreover, 
It is shown that an ideal convergence rate without interference can be asymptotically achieved by increasing $N$. 
Numerical results are conducted to demonstrate the effectiveness of our proposed schemes and verify the analytical results in a variety of FL settings.
\end{itemize}

The rest of this paper is organized as follows. The models and objectives are presentd in Sections~\Rmnum{2} and~\Rmnum{3}, respectively. Section~\Rmnum{4} introduces the concept of phase-manipulated favorable propagation and channel hardening enabled by RIS, and proposes two robust aggregation schemes. In Section~\Rmnum{5}, we analyze the MSE and convergence of the proposed schemes. Simulation results and conclusions are in Sections \Rmnum{6}~and~\Rmnum{7}, respectively.

Throughout the paper, numbers, vectors, and matrices are represented by lower-case, boldface lower-case, and boldface uppercase letters, respectively. The operator $|\!\cdot\!|$ returns the absolute value of a complex number. If used with a set, $|\!\cdot\!|$ returns the cardinality of the set. The operator $\left\|\cdot\right\|$ returns the Euclidean norm of a vector. Let $\mathbb{R}$ and $\mathbb{C}$ denote the set of real and complex numbers, respectively. We use $\mathbb{E}[\cdot]$ and $\mathbb{V}[\cdot]$ to denote the expectation and variance of a random variable (RV), respectively. The operators $\Re\{\cdot\}$, $\Im\{\cdot\}$, and $\angle$ return the real part, imaginary part, and phase of a complex number, respectively. The superscripts $(\cdot)^T$, $(\cdot)^\ast$, and $(\cdot)^H$ stand for the transpose, conjugate, and conjugate-transpose operations, respectively. The symbol $\mathcal{CN}(\mathbf{x},\mathbf{\Sigma})$ is a circularly symmetric complex Gaussian distribution with mean $\mathbf{x}$ and covariance $\mathbf{\Sigma}$, $\mathrm{Exp}(\lambda)$ is the exponential distribution with rate parameter $\lambda$, and $\mathcal{U}(a,b)$ is a uniform distribution between $a$ and $b$.

% \vspace{-8.pt}
\section{System Model}
We consider an AirFL system as illustrated in Fig.~\ref{fig1} that comprises a central PS and $K$ target devices. 
The AirComp process is perturbed by external interference, e.g., from other cells, tasks, and attackers. 
The learning and communication models are described separately in the following.

\subsection{Learning Model}
\begin{figure}[!t]
\centering
\includegraphics[width=3.2in]{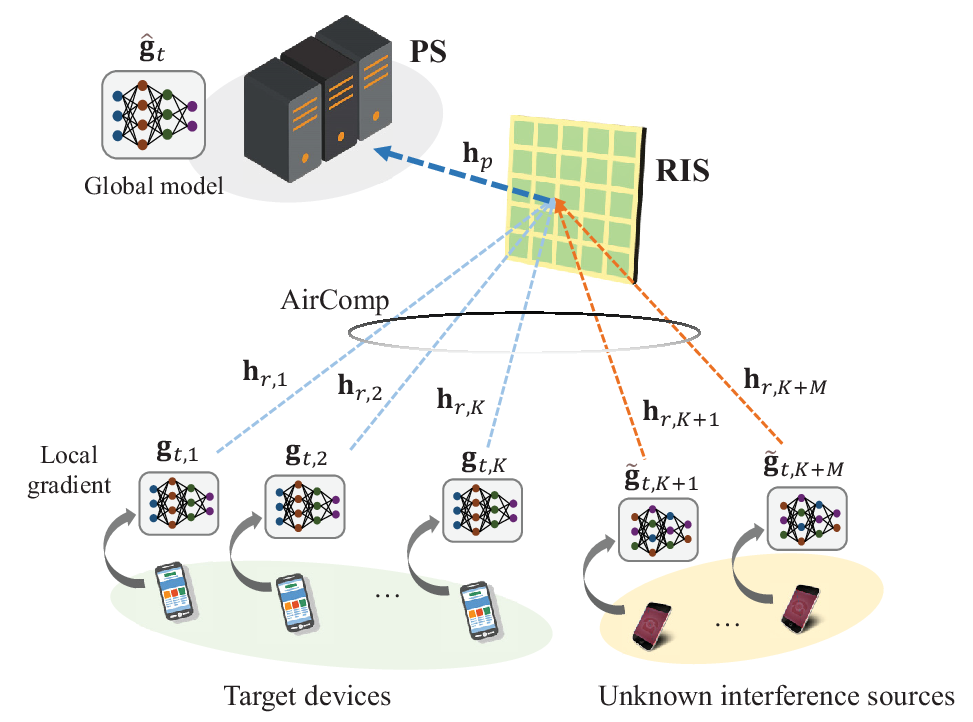}
\caption{The framework of an RIS-enhanced AirFL system with interference.}
% \vspace{-0.4cm}
\label{fig1}
\end{figure}
We first describe the FL process underpinning AirFL. 
% Considering a typical FL system comprising a central PS and $K$ target devices.
% The PS coordinates several target devices to train a shared machine learning model. 
Each target device $k\!\in\!\mathcal{K}\!\triangleq\!\!\{1,2,\cdots,K\}$ owns its local dataset $\mathcal{D}_k$ with $\vert \mathcal{D}_k \vert$ training samples. 
% Moreover, it is worth noting that the data distributions of different devices are non-IID in practice, which poses additional challenges to model training. 
% Hence, the local loss function of model parameters, denoted by $\mathbf{w}\in \mathbb{R}^d$, varies over different devices. 
The local loss function at device~$k$ is defined as 
\begin{align}
F_k (\mathbf{w},\mathcal{D}_k)= \frac{1}{\vert \mathcal{D}_k \vert} \sum_{\mathbf{u}\in \mathcal{D}_k}  \mathcal{L}(\mathbf{w},\mathbf{u}),
\end{align}
where $\mathbf{w}\!\in\!\mathbb{R}^D$ is the $D$-dimensional model parameter vector, $\mathbf{u}$ is the data sample selected from $\mathcal{D}_k$, and $\mathcal{L}(\mathbf{w},\mathbf{u})$ represents the sample-wise loss function.
Without loss of generality, we assume that all local datasets have a uniform size, i.e., $\vert\mathcal{D}_k\vert=D$, $\forall k\in \mathcal{\mathcal{K}}$.~\footnote{We assume a uniform size for all client weights to simplify notation and extract insightful observations, which aligns with similar setups in some existing schemes, e.g., \cite{8,11,RIS-AirFL-2,RIS-AirFL-3,RIS-AirFL-4}. In fact, our proposed aggregation scheme is extendable to scenarios with unbalanced aggregation weights \cite{FedNova}.}
The learning process aims to optimize the model parameter $\mathbf{w}$ to minimize the global loss function defined as
\begin{align}\label{eq1}
 F(\mathbf{w})=\frac{1}{K}\sum_{k\in\mathcal{K}} F_k(\mathbf{w},\mathcal{D}_k).
\end{align}

%We assume that the PS owns the knowledge of the trusted devices but cannot obtain any information about the possible Byzantine attackers. 
Distributed stochastic gradient descent (SGD) is adopted to minimize $F(\mathbf{w})$, which optimizes $\mathbf{w}$ in an iterative manner. 
% We first introduce the ideal training process without transmission distortions and interference. 
Specifically, the $t$-th round of model training is made up of the following steps:
\begin{itemize}
\item [1)]\emph{Model broadcasting}: The PS broadcasts the latest global model $\mathbf{w}_t$ to all devices. 
\item [2)]\emph{Local computing}: Based on the received global model $\mathbf{w}_t$, each target device computes its local gradient based on a local mini-batch $\mathcal{B}_{t,k}$ of size $b_k$, which is expressed as
\begin{align}\label{eq3}
\mathbf{g}_{t,k}\triangleq \nabla F_k(\mathbf{w}_{t},\mathcal{B}_{t,k})=\frac{1}{\vert\mathcal{B}_{t,k} \vert} \sum_{\mathbf{u}\in \mathcal{B}_{t,k}}   \nabla \mathcal{L}(\mathbf{w}_t,\mathbf{u}),
\end{align}
where $\mathcal{B}_{t,k}$ is selected from the local dataset $\mathcal{D}_k$.
\item [3)]\emph{Local updates uploading}: Each target device reports its local gradient, $\mathbf{g}_{t,k}\in\mathbb{R}^{D}$, to the PS, which Euclidean norm, $\Vert\mathbf{g}_{t,k}\Vert$, is upper bounded by a finite constant $G$.
% \textcolor{blue}{$\Vert\mathbf{g}_{t,k}\Vert\le G$}.
\item [4)]\emph{Model aggregation}: Upon receiving all the local gradients, the PS calculates the global gradient
\begin{align}\label{gradient}
\mathbf{g}_t = \frac{1}{K} \sum_{k=1}^K \mathbf{g}_{t,k},
\end{align}
and updates the global model according to
\begin{align}
\mathbf{w}_{t+1}=\mathbf{w}_t-\eta_t \mathbf{g}_t,
\end{align}
where $\eta_t$ is a chosen learning rate at $t$-th training round.
\end{itemize}
% The above steps iterate until a convergence condition is met.
% Iterating the above steps until the FL algorithm converges and a common FL model is available. 

% \vspace{-15.pt}
\subsection{Communication Model}
AirComp is adopted to realize efficient uploading and model aggregation in Fig.~\ref{fig1}. In AirComp, participating devices simultaneously upload the analog signals of local gradients to the PS, and hence a weighted summation of the local updates in (\ref{gradient}) is achieved by exploiting the waveform superposition nature of wireless channels. Since the analog aggregation of AirComp is vulnerable to interference, which includes both unintentional interference (e.g., inter-task/cell interference) and malicious attacks, we introduce the RIS technology to combat interference. The interference resilience of RIS is based on the concept of phase-manipulated favorable propagation discussed in Section~\Rmnum{4}. Here, we depict the basic framework of this RIS-empowered AirFL system.

As shown in Fig. \ref{fig1}, an RIS with $N$ reflecting elements is deployed to assist the AirFL system against interference. To facilitate analysis, we assume that there are $M$ unknown interference sources and denote the set of them by $\mathcal{M}\!\triangleq\!\{K\!+\!1,\cdots, K\!+\!M\}$. Then, within this AirFL framework, the received signal at the PS in the $t$-th round is expressed as
\begin{align}\label{received1}
\mathbf{y}_t=\sum_{k\in\mathcal{K}} h_k^E  \sqrt{p_k} \mathbf{g}_{t,k} +\sum_{m\in \mathcal{M}} h_m^E \sqrt{p_m} \tilde{\mathbf{g}}_{t,m}+\mathbf{z}_t,
\end{align}
where $p_i$ is the transmit power for device $i\in \mathcal{N}\!\triangleq\!\mathcal{K} \cup \mathcal{M}$, $\mathbf{z}_t$ is additive white Gaussian noise $\mathcal{CN}(\mathbf{0},\sigma^2\mathbf{I}_{D})$, and $\tilde{\mathbf{g}}_{t,m}\in\mathbb{R}^{D}$ denotes the signal vector transmitted by the interferer $m$. The interference signal is assumed to be arbitrary values with normalized power, i.e., $\Vert\tilde{\mathbf{g}}_{t,m}\Vert=1$. The cascaded channel from device $i$ to the PS through the RIS, $h_i^{E}$, is given by 
\begin{align}
h_i^E=\beta_i\mathbf{h}_p^H\mathbf{\Theta}\mathbf{h}_{r,i}, \enspace \forall i \in \mathcal{N},
\end{align}
where $\beta_i$ denotes the equivalent large-scale fading coefficient, which represents the product of the large-scale fading coefficients of the RIS-PS and device $i$-RIS links, $\mathbf{h}_p\!\sim\!\mathcal{CN}(\mathbf{0},\mathbf{I}_{N})$ and $\mathbf{h}_{r,i}\!\sim\!\mathcal{CN}(\mathbf{0},\mathbf{I}_{N})$ denote the small-scale fading channel from the RIS to the PS and device $i$ to the RIS, respectively, $\mathbf{\Theta}\triangleq {\rm diag}\left\{{\rm e}^{j\theta_1},\ldots,{\rm e}^{j\theta_n},\ldots,{\rm e}^{j\theta_N}\right\}$ is the reflection matrix of the RIS, and $\theta_n\in[0,2\pi)$ is the phase shift introduced by the $n$-th RIS reflecting element, which is set to be invariant within a communication round \cite{RIS-AirFL-3}. The channel coefficients of $\mathbf{h}_p$ and $\mathbf{h}_{r,k}$ are assumed to be perfectly known to the PS. 
% Extensive approaches have been proposed for channel estimation of RIS-aided links \cite{32,33}.
For practical consideration, we assume that the PS cannot acquire any knowledge of $\mathbf{h}_{r,m}$.
Meanwhile, assuming that direct links between the PS and edge devices deployed in coverage-challenged areas are heavily obstructed by trees, buildings, and other environmental factors, we neglect these direct links due to their comparatively lower channel gain compared to RIS-related channels. This assumption is commonly adopted in typical RIS-assisted communication scenarios \cite{link,link1,link2}.

Based on the received signal in (\ref{received1}), the PS computes an estimated global gradient as
\begin{align} \label{ghata}
\hat{\mathbf{g}}_t\!=\!\frac{\Re\!\left\{\mathbf{y}_t\right\}}{\lambda}\!=\!\sum_{k\in\mathcal{K}}\!{{\ell_k}\mathbf{g}_{t,k}}\!+\!\!\sum_{m\in\mathcal{M}}\!{{\ell_m}{\tilde{\mathbf{g}}}_{t,m}}\!+\!\bar{\mathbf{z}}_t,
\end{align}
where $\lambda$ is a denoising factor introduced by the PS, $\bar{\mathbf{z}}_t\triangleq \frac{\Re\{\mathbf{z}_t\}}{\lambda}$ is the equivalent noise, and the aggregation coefficient ${\ell_k}$ and interference coefficient ${\ell_m}$ are, respectively, expressed as
\begin{align} \label{aggregation}
{\ell_k}\triangleq\frac{\beta_k\sqrt{p_k}}{\lambda}\Re\!\left\{\mathbf{h}_p^H\mathbf{\Theta}\mathbf{h}_{r,k}\!\right\},
\end{align}
and
\begin{align} \label{interference}
{\ell_m}\triangleq\frac{\beta_m\sqrt{p_m}}{\lambda}\Re\!\left\{\mathbf{h}_p^H\mathbf{\Theta}\mathbf{h}_{r,m}\!\right\}.
\end{align}
% where the transmit power $p_k$ and RIS phase shifts $\mathbf{\Theta}$ are adjustable.

\begin{figure*}[!t]
\centering
\includegraphics[width=6.8in]{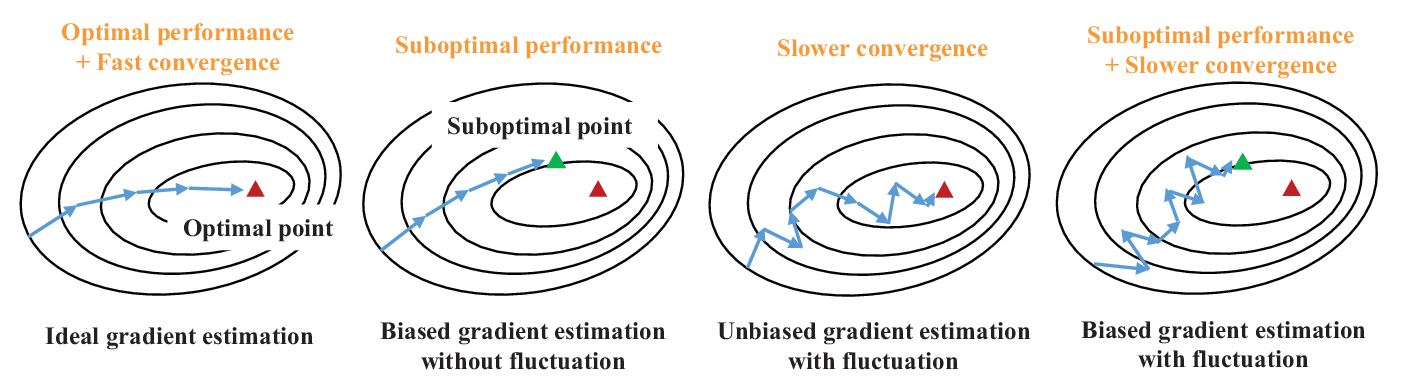}
\caption{The effects of long-term bias and instantaneous random fluctuation on the convergence performance and rate.}
\label{fig1-1}
\end{figure*}

By comparing (\ref{ghata}) and (\ref{gradient}), we see that interference affects the performance of gradient estimation in two aspects:
\begin{itemize}
    \item First, from a \emph{long-term statistical perspective}, the interference coefficient ${\ell_m}$ introduces bias into the gradient estimation.
    \item Second, from an \emph{instantaneous perspective}, the existence of interference intensifies random fluctuations of gradient estimation, which increases the error variance.
\end{itemize}
As a result, both the \emph{long-term bias} and \emph{instantaneous random fluctuations} caused by interference can adversely affect the AirFL convergence. In this paper, our primary objective for interference suppression is to minimize its detrimental effect on signal quality, thereby enhancing the overall reliability and convergence of the model. While conventional communication systems use metrics like the signal-to-interference-plus-noise ratio (SINR) to evaluate signal quality, such metrics do not directly correlate with FL convergence, leading to suboptimal designs. Therefore, we begin by establishing performance objectives tailored for interference suppression in wireless FL systems.

\section{Design of Performance Objectives}

In this section, we investigate the design of performance objectives aimed at minimizing the impact of interference in FL tasks. The most straightforward approach would be to minimize the global loss function in (\ref{eq1}). However, accurately characterizing the loss function under communication errors in a closed form remains an open challenge, making direct optimization infeasible. Alternative studies, such as \cite{11}, rely on theoretical convergence analysis to establish an upper bound for the loss function, subsequently adopting this bound as the performance objective. Numerical analyses in \cite{5,7} show that the derived upper bound is not strictly tight and can exhibit notable deviation from the true loss function, posing challenges in ensuring optimal convergence performance.

To theoretically characterize the impact of interference on FL convergence and design targeted performance objectives, we begin with a universal performance analysis framework suitable for scenarios involving imperfect gradient estimations. Specifically, we denote the estimated global gradient in the $t$-th round in a general form, i.e.,
\begin{align}
    \hat{\mathbf{g}}_t=\mathbf{g}_t +\bm{\varepsilon}_t,
\end{align}
where $\bm{\varepsilon}_t$ represents the random estimation error, which accounts for all possible sources of error, including misalignment, interference, and additive noise. We denote its statistical properties by defining $\mathbf{e}_1\triangleq \mathbb{E}\left[\bm{\varepsilon}_t\right]$ and  $e_2\triangleq \mathbb{E}\left[\Vert \bm{\varepsilon}_t\Vert^2\right]$. 
% Then, according to  \cite[\emph{Theorem 1}]{11}, we conclude the following key observations regarding the FL convergence behavior:
Since our focus is on transmission design at the physical layer, we do not consider optimization at the algorithmic level. As such, the only tunable parameters are $\mathbf{e}_1$ and $e_2$, while other system-level factors, such as data heterogeneity, are kept fixed. This ensures that the proposed transmission method can be applied to a broad range of learning scenarios.
Then, according to \cite[\emph{Theorem 1}]{11}, we conclude the following key observations regarding the impact of $\mathbf{e}_1$ and $e_2$ on FL convergence behavior:
\begin{itemize}
    \item The ultimate performance of the FL algorithm after convergence primarily depends on the first-order moment of the gradient estimation error, $\mathbf{e}_1$. With a sufficiently small learning rate, global optimum convergence is asymptotically achievable if $\mathbf{e}_1=\mathbf{0}$, indicating an unbiased gradient estimation. Otherwise, the algorithm converges to a biased local optimum. This observation is also supported by the findings in \cite{biased-ICML}.
    \item For a fixed $\mathbf{e}_1$, a smaller second-order moment $e_2$, which represents the MSE of gradient estimation, can accelerate convergence and help approach the optimal point.
\end{itemize}

Fig.~\ref{fig1-1} provides an intuitive illustration of these two observations. From a qualitative perspective, the long-term bias can lead to an astray model after the training, ultimately degrading AirFL’s convergence performance, especially in the presence of non-independent and identically distributed (non-IID) local datasets~\cite{35}. Additionally, instantaneous random fluctuations introduce uncertainty into each update step, which hinders the gradient descent process and leads to slower convergence rates.

Building upon these observations, it is evident that the statistics $\mathbf{e}_1$ and $e_2$ are critical factors affecting the FL convergence. Consequently, to effectively mitigate the impact of interference in FL tasks, our goal is to achieve \textbf{unbiasedness} of gradient estimation to seek optimality and a \textbf{minimum possible MSE} is also desirable to speed up convergence.
Specific definitions of these two performance objectives are as follows.

\textbf{1) Unbiasedness}: According to \cite[Lemma~1]{35}, the expectation of the estimated global gradient $\hat{\mathbf{g}}_t$ in (\ref{ghata}) is equal to the ground-truth global gradient $\mathbf{g}_t$ defined in (\ref{gradient}), i.e., $\mathbb{E}\left[\hat{\mathbf{g}}_t\right]=\mathbf{g}_t$, which is equivalent to
\begin{align}\label{Unbiasedness}
    \mathbb{E}\left[{\ell_k}\right]=\frac{1}{K},~~\mathbb{E}\left[{\ell_m}\right]=0,
\end{align}
where the expectations are taken over the distributions of~channel fadings $\mathbf{h}_p$ and $\mathbf{h}_{r,i}$. Therefore, to achieve unbiasedness, the imbalance of aggregation coefficients $\{\ell_k\}_{k\in \mathcal{K}}$, caused by heterogeneous large-scale fading coefficient $\beta_k$ and random small-scale fading channel $\mathbf{h}_{r,k}$, and the interference coefficients $\{\ell_m\}_{m\in \mathcal{M}}$ should be statistically eliminated.

\textbf{2) Minimum possible MSE}: The MSE of gradient estimation \cite{Eldar3}, given by
\begin{align}\label{MSE}
{\rm MSE}=\mathbb{E}\left[\Vert{\hat{\mathbf{g}}}_t-\mathbf{g}_t\Vert^2\right],
\end{align}
is expected to be as small as possible to expedite convergence, where the expectation is taken over the distributions of channel fading, local gradient $\mathbf{g}_{t,k}$, interference $\tilde{\mathbf{g}}_{t,m}$, and noise $\mathbf{z}_t$.

By meeting these two objectives, we enable a rigorous theoretical convergence analysis, as detailed in \cite{11}, which shows that the FL algorithm can achieve rapid convergence to the optimal point under a sufficiently small learning rate. In the following section, we explore how to accomplish these objectives through the joint design of RIS phase shifts and transceiver signal processing.

% \vspace{-3.pt}
\section{RIS-Empowered Robust Aggregation for AirFL}
In order to attain unbiased estimation with minimized MSE, we first introduce the concept of phase-manipulated favorable propagation and channel hardening enabled by RIS. Based on this concept, we develop two robust aggregation schemes with power allocation and RIS phase shift settings for AirFL.

% \vspace{-5.pt}
\subsection{Phase-Manipulated Favorable Propagation and Channel Hardening for AirComp}
Recall that in conventional massive MIMO systems, the asymptotic vector-wise orthogonality among different wireless channel vectors provides favorable propagation, and the asymptotic element-wise orthogonality of the channel vector ensures channel hardening \cite{favorable}. These properties align seamlessly with the requirements of robust aggregation for AirFL, aiming to filter out unwanted interference signals and diminish the error variance, respectively. Recent study has utilized the properties of favorable propagation and channel hardening in massive MIMO for realizing AirFL \cite{36}.

Nevertheless, the enhancement in AirFL computational performance enabled by the massive MIMO comes at the expense of requiring large-scale receiving antennas, leading to a significant escalation in hardware cost. Moreover, unlike traditional uplink transmission, AirComp for computation tasks is mostly single-stream transmission. Thus, introducing additional radio frequency (RF) links merely enhances diversity gains, resulting in superfluous utilization. Hence, we propose to replace costly large-scale antennas with lower-cost RIS, which fortunately demonstrates that same functions as large-scale antennas are attained in the AirFL through phase manipulation at the RIS. 
\begin{theorem} \label{theorem1}
We set the RIS phase shifts as     
    \begin{align}\label{eq10}
    \theta_n=-\angle h_{p,n}^\ast+\angle\sum_{k\in \mathcal{K}}w_k h_{r,k,n}^\ast, \enspace \forall n=1,\cdots,N,
    \end{align}
where $w_k\!>\!0$ is an arbitrary weight factor for device $k$, and $h_{p,n}$ and $h_{r,k,n}$ are the $n$-th elements of channel vectors $\mathbf{h}_p$ and $\mathbf{h}_{r,k}$, respectively. This setting preserves the signal from target devices and achieves statistical interference elimination, accomplishing favorable propagation. In particular, it yields
\begin{align}\label{th1_1}
    \mathbb{E}\left[{\ell_k}\right]= \frac{\pi N \beta_k\sqrt{p_k}w_k}{4\lambda\sqrt{\sum_{i=1}^{K}w_i^2}},\enspace \mathbb{E}\left[{\ell_m}\right]=0.
\end{align}
\end{theorem}
\begin{proof}
    See Appendix \ref{appB}. $\hfill\blacksquare$
\end{proof}

From (\ref{th1_1}), it is intuitive to linearly scale up the denoising factor, $\lambda$, with $N$, to attain unbiased estimation in (\ref{Unbiasedness}). Under such a setting for $\lambda$, we further verify in the following theorem that a large-scale RIS also induces channel hardening.

\begin{theorem} \label{theorem2}
When the RIS phase shifts are configured according to (\ref{eq10}) and $\lambda$ scales linearly with $N$, both variances of the aggregation coefficient $\ell_k$ and interference coefficient $\ell_m$ diminish by the order of $\mathcal{O}\left(\frac{1}{N}\right)$.
\end{theorem}

\begin{proof}
    See Appendix \ref{appC}. $\hfill\blacksquare$
\end{proof}

As $N\to \infty$, the variances of $\ell_k$ and $\ell_m$ tend towards zero, i.e., the aggregation coefficient $\ell_k$ and interference coefficient $\ell_m$ are approximated as constants devoid of any fluctuations. It implies that the channel hardening effect, typically achieved through costly extensive antennas, is also attainable by setting the RIS phase shifts in (\ref{eq10}) for low-cost RIS elements.

In general, by leveraging the favorable propagation and channel hardening, we can anticipate realizing the~objectives in Section~\Rmnum{3} through the following aggregation schemes. 

% \vspace{-10.pt}
\subsection{Proposed Robust Aggregation Schemes}

By exploiting the RIS phase shifts in \emph{Theorem~\ref{theorem1}}, we statistically eliminate interference. In addition, it is imperative to handle the imbalanced aggregation coefficients $\ell_k$, ensuring the unbiasedness in~(\ref{Unbiasedness}). Observing (\ref{th1_1}), we seek to find the transmit power $p_k$, weight factor $w_k$, and denoising factor $\lambda$ so that the following equality is ensured.
\begin{align}\label{eq15}
    \mathbb{E}\left[{\ell_k}\right]= \frac{\pi N \beta_k\sqrt{p_k}w_k}{4\lambda\sqrt{\sum_{i=1}^{K}w_i^2}}=\frac{1}{K}, \enspace \forall k \in \mathcal{K}.
\end{align}
The imbalance stems from the heterogeneity of large-scale fading coefficients, $\{\beta_k\}_{k\in \mathcal{K}}$, necessitating its offset by adjusting $p_k$ and $w_k$. In view of this, we develop the following schemes.~\footnote{Any combination of transmit power $p_k$ and weight factor $w_k$ that satisfies (\ref{eq15}) constitutes a valid transmission scheme, effectively eliminating the large-scale heterogeneity introduced by $\beta_k$ at the first-order moment. In this work, we propose two representative schemes that separately design $w_k$ and $p_k$, achieving performance advantages in gradient computation and interference suppression, respectively, as demonstrated in the following section.}

\subsubsection{Proposed Transmission  Scheme \Rmnum{1}}
In this scheme, we adjust $p_k$ to eliminate the large-scale heterogeneity. The transmit power $p_k$ is set to $\sqrt{p_k}\!=\!\beta_k^{-1}\zeta$, where $\zeta$ is a scaling factor to satisfy the transmit power constraint, i.e., $p_k\Vert\mathbf{g}_{t,k}\Vert^2\!\le\! P_k$. This translates to $\zeta\!=\!\frac{\min\limits_{k\in\mathcal{K}}{\sqrt{P_k}\beta_k}}{G}$, where $P_k$ represents the maximum transmit power for device $k\!\!\in\!\!\mathcal{K}$. This setting resembles the idea of channel inversion~\cite{8}, but we only invert the large-scale coefficients without the necessity for instantaneous CSIT.

Given that the heterogeneity of $\{\beta_k\}_{k\in \mathcal{K}}$ in (\ref{eq15}) is entirely eliminated, we have the following proposition.

\begin{proposition} \label{proposition2}
By setting $\sqrt{p_k}\!=\!\beta_k^{-1}\zeta$, the denoising factor $\lambda=\frac{\pi N\sqrt K\min\limits_{k\in\mathcal{K}}{\sqrt{P_k}\beta_k}}{4G}$, and $w_k=1$, i.e., the RIS phase shifts
    \begin{align}\label{RISphase1}
    \theta_n=-\angle h_{p,n}^\ast+\angle\sum_{k=1}^{K}h_{r,k,n}^\ast,\enspace \forall n=1,\cdots,N,
    \end{align}
the gradient estimation in (\ref{ghata}) at the PS is unbiased.
\end{proposition}

\begin{proof}
By directly applying \emph{Theorem \ref{theorem1}} and substituting the parameters into (\ref{th1_1}), the proof is completed.
$\hfill\blacksquare$
\end{proof}

\subsubsection{Proposed Transmission Scheme \Rmnum{2}}
While Scheme \Rmnum{1} eliminates the heterogeneity of $\{\beta_k\}_{k\in \mathcal{K}}$ through power control, the power efficiency can be severely sacrificed. 
An alternative approach is to adjust the weight factors $\{w_k\}_{k\in \mathcal{K}}$.
Assume that each device $k$ utilizes its maximum power $P_k$ for transmission. For the sake of tractability, we substitute $\Vert\mathbf{g}_{t,k}\Vert$ with its uppper bound $G$ and get $p_k=\frac{P_k}{G^2}$ for all $k\in \mathcal{K}$.% Note that this setting not only ensures the maximum power constraint, but also has no impact on the final analytical results., i.e., $p_k\Vert\mathbf{g}_{t,k}\Vert^2\!=\!P_k$ 

The configuration of RIS phase shifts of Scheme~\Rmnum{2} is designed in the following proposition.

\begin{proposition} \label{proposition3}
By setting $\sqrt{p_k}=\frac{\sqrt{P_k}}{G}$, the denoising factor $\lambda=\frac{\pi NK}{4G\sqrt{\sum_{k=1}^{K}w_k^{2}}}$, and $w_k=\frac{1}{\beta_k\sqrt{P_k}}$, i.e., the RIS phase shifts
    \begin{align}\label{RISphase2}
    \theta_n\!=\!-\angle h_{p,n}^\ast\!+\!\angle\sum_{k=1}^{K}{\frac{1}{\beta_k\sqrt{P_k}}h_{r,k,n}^\ast},\enspace \forall n=1,\cdots,N,
    \end{align}
the gradient estimation in (\ref{ghata}) at the PS is unbiased.
\end{proposition}

\begin{proof} By directly applying \emph{Theorem \ref{theorem1}} and substituting the parameters into (\ref{th1_1}), the proof is completed.
$\hfill\blacksquare$
\end{proof}

Both Scheme~\Rmnum{1} and Scheme \Rmnum{2} provide an unbiased gradient estimation. However, their impacts on the MSE of the gradient estimation may differ due to distinct power allocation strategies and RIS phase shift configurations. Thorough analyses are presented in the subsequential section.

% \vspace{-2.pt}
\section{Performance Analysis}
In this section, the MSE of the gradient estimation for the proposed schemes is accurately derived in closed form, which facilitates the convergence analysis for further evaluation.

% \vspace{-5.pt}
\subsection{MSE Analysis}
The MSE, defined in (\ref{MSE}), quantifies the AirComp performance for global model aggregation in AirFL. 
% By substituting (\ref{gradient}) and (\ref{ghata}) into (\ref{MSE}), it is easy to get~(\ref{MSEnew}), where ${\bar{h}}_k\triangleq\ell_k-\frac{1}{K}$ and ${\bar{h}}_m\triangleq\ell_m$.
Concerning an interferer $m\in \mathcal{M}$, we consider the worst case that it transmits signal at the maximum power $P_m$, i.e., $p_m = P_m$.
For the MSE analysis, we present the following theorem.

\begin{theorem}\label{theorem3}
The MSE for Scheme~\Rmnum{1} and Scheme \Rmnum{2}, denoted by ${\rm MSE}_1$ and ${\rm MSE}_2$, is calculated in (\ref{MSE1}) and (\ref{MSE2}), respectively, where $\alpha_k\triangleq \beta_k\sqrt{P_k}$ and $\alpha_m\triangleq \beta_m\sqrt{P_m}$.
\end{theorem}

\begin{proof} See Appendix~\ref{appD}.$\hfill\blacksquare$
\end{proof}

\begin{figure*}[t] 
\begin{align}
{\rm MSE}_1\!=&\!\underbrace{\frac{8\left(K\!+\!1\right)\!-\!\pi^2}{\pi^2NK^2}\sum_{k\in\mathcal{K}}\mathbb{E}\left[\Vert\mathbf{g}_{t,k}\Vert^2\right]\!+\!\frac{8-\pi^2}{\pi^2NK^2}\sum_{k\in\mathcal{K}}\sum_{k^\prime\neq k}\mathbb{E}\!\left[\mathbf{g}_{t,k}^T\mathbf{g}_{t,k^\prime}\right]}_{\rm{computation},\ \Delta_{1,1}}\!+\!\underbrace{\sum_{m\in\mathcal{M}}\frac{8G^2\alpha_m^2}{\pi^2NK\min\limits_{k\in\mathcal{K}}{\alpha_k^2}}}_{\rm{interference},\ \Delta_{1,2}}+\underbrace{\frac{8G^2\sigma^2D}{\pi^2N^2K\min\limits_{k\in\mathcal{K}}{\alpha_k^2}}}_{\rm{noise},\ \Delta_{1,3}}.
\label{MSE1}
\end{align}
\hrulefill
\end{figure*}

\begin{figure*}[t] 
\begin{align}
{\rm MSE}_2\!=\!\underbrace{\sum_{k\in\mathcal{K}}\!\frac{8\!\left(\frac{\sum_{i\in\!\mathcal{K}}\alpha_i^{-2}}{\alpha_k^{-2}}\!+\!1\right)\!-\!\pi^2}{\pi^2NK^2}\mathbb{E}\!\left[\Vert\mathbf{g}_{t,k}\Vert^2\right]\!+\!\frac{8-\pi^2}{\pi^2NK^2}\sum_{k\in\mathcal{K}}\sum_{k^\prime\neq k}\mathbb{E}\!\left[\mathbf{g}_{t,k}^T\mathbf{g}_{t,k^\prime}\right]}_{\rm{computation},\ \Delta_{2,1}}\!+\!\underbrace{\sum_{m\in\mathcal{M}}\frac{8G^2\sum_{i\in\mathcal{K}}\alpha_i^{-2}}{\pi^2NK^2\alpha_m^{-2}}}_{\rm{interference},\ \Delta_{2,2}}\!+\!\underbrace{\frac{8G^2\sigma^2D\sum_{i\in\mathcal{K}}\alpha_i^{-2}}{\pi^2N^2K^2}}_{\rm{noise},\ \Delta_{2,3}}.
\label{MSE2}
\end{align}
\hrulefill
\end{figure*}

To facilitate the analysis, we divide the derived MSE into three parts, i.e., computation errors $\Delta_{1,1}$ and $\Delta_{2,1}$ (self- and cross-correlation terms of local gradients), interference errors $\Delta_{1,2}$ and $\Delta_{2,2}$ (self-correlation terms of interference signals), and noise errors $\Delta_{1,3}$ and $\Delta_{2,3}$ (equivalent noise power). 
In this way, we gain some insights~into the effectiveness of the proposed schemes in gradient computation, interference resilience, and noise suppression. 

\vspace{-5.pt}
\begin{remark}[Impact of the number of RIS elements]
    Both the computation and interference errors diminish by an order of $\mathcal{O}\left(\frac{1}{N}\right)$. Concurrently, the noise error showcases a decrease of order $\mathcal{O}\left(\frac{1}{N^2}\right)$, which is not as dominant a factor in determining the MSE. Furthermore, the MSE tends to zero as $N$ increases, which implies fast convergence.
\end{remark}

% \vspace{-12.pt}
\begin{remark}[Impact of SNR]
    With increasing SNR $\frac{P_k}{\sigma^2}$, the impact of interference and noise diminishes, which is consistent with established results in pure communication scenarios. However, the idea of eliminating interference error by increasing only the useful signal power $P_k$ is not cost-effective, given that the interference power $P_m$ typically remains significant. Furthermore, the computational error does not decrease when $P_k$ grows, leading to the MSE converging to a non-zero constant rather than approaching zero. Consequently, enhancing transmit power is less effective than increasing the number of RIS reflecting elements for improving the MSE.
\end{remark}

% \vspace{-5.pt}
In addition, we undertake a comparative evaluation to identify the preferable applicable scenarios for each scheme.

\emph{Observation 1:} 
In terms of gradient computation, empirical data from simulations, assuming common distributions for $\mathbb{E}\left[\Vert \mathbf{g}_{t,k}\Vert^2 \right]$, suggests that Scheme~\Rmnum{1} outperforms Scheme~\Rmnum{2} in terms of computational performance with high probability. This observation is further supported by the simulation results in Section~\Rmnum{6}. Especially, for IID local datasets where $\mathbb{E}\left[\Vert \mathbf{g}_{t,k}\Vert^2 \right]$ is consistent for all $k$, the computation superiority of Scheme~\Rmnum{1} is rigorously proven by leveraging the Cauchy-Schwarz inequality. This result is attributed to their distinctions in handling the large-scale coefficient $\{\beta_k\}_{k\in \mathcal{K}}$. Specifically, Scheme~\Rmnum{1} directly eliminates $\beta_k$ via power control at the transmitter, fundamentally addressing the issue of large-scale heterogeneity and resulting in improved gradient computation. In contrast, Scheme II enables full power transmission, offering better interference suppression, but it only addresses large-scale heterogeneity at the first-order moment through the settings of $w_k$. This approach does not eliminate heterogeneity at the second-order moment, as shown in $\Delta_{2,1}$. This incomplete strategy for handling $\beta_k$ affects the accuracy and stability of computational performance, leading to higher MSE in certain circumstances.

\emph{Observation 2:} 
In terms of interference and noise suppression, Scheme \Rmnum{2} always achieves better performance than Scheme \Rmnum{1} due to the fact that $\frac{\sum_{i\in\mathcal{K}} \alpha_i^{-2}}{K}\leq \max_{k\in\mathcal{K}}\alpha_k^{-2}$.
This is owing to Scheme~\Rmnum{2}'s strategy that target devices employ full power transmission, which consequently results in more effective suppression of interference and noise than Scheme~\Rmnum{1}.

Note that further optimization of $w_k$ and $p_k$ based on the statistical distribution of $\mathbf{g}_{t,k}$ could strike a more effective trade-off between gradient computation and interference suppression, potentially enhancing overall MSE performance beyond our current schemes. However, this approach relies on additional approximations of gradient statistics, which requires further investigation.

In summary, we conclude that \emph{Scheme~\Rmnum{1} excels in the gradient computation, making it more suitable for \textbf{computation-dominant systems}, while Scheme~\Rmnum{2} focuses more on combating interference and noise, thus more suitable for \textbf{interference-dominant systems}}.

% \vspace{-12.pt}
\subsection{Convergence Analysis of AirFL}
To begin with, we need some common assumptions on loss functions, which have been widely used \cite{4,5,wireless1,adapt}.

\emph{Assumption 1}: The local loss functions $F_k(\cdot)$ are differentiable and have $L$-Lipschitz gradients, which follows
\begin{align}
F_k(\mathbf{w})\leq  F_k(\mathbf{v}) +\nabla F_k(\mathbf{v})^T (\mathbf{w}-\mathbf{v})+\frac{L}{2}\Vert \mathbf{w}-\mathbf{v}\Vert^2.
\end{align}

\emph{Assumption 2}: The stochastic gradient is unbiased and variance-bounded, i.e.,
$\mathbb{E}\!\left [ \mathbf{g}_{t,k} \right]\!=\!\nabla F_k (\mathbf{w}_t)$ and $\mathbb{V}\left [ \mathbf{g}_{t,k} \right] \leq \chi^2$.

\emph{Assumption 3}: The gradient dissimilarity between the local and global gradients is bounded by a finite value $\xi$, i.e., $\left \Vert \nabla F_k(\mathbf{w})\right \Vert\leq \xi \left \Vert \nabla F(\mathbf{w})\right \Vert$.

It is worth noting that $\xi$ increases with the level of data heterogeneity and $\xi=1$ corresponds to the ideal case with IID local datasets \cite{adapt}. Based on the above assumptions, we evaluate the FL convergence under the proposed RIS-aided robust aggregation schemes in the following theorem.

\begin{theorem}\label{theorem4}
    Suppose the learning rate $\eta_t=\frac{1}{\varpi_u\sqrt{T}}$. The convergence of AirFL at the $T$-th round is bounded by
    \begin{align}\label{th4-1}
        \frac{1}{T}\!\sum_{t=0}^{T-1} \mathbb{E}\!\left[\left\Vert \nabla F(\mathbf{w}_t) \right \Vert^2 \right]\! \leq \!\frac{2\varpi_{u}}{\sqrt{T}}\!\left(\!F(\mathbf{w}_0)\!-\!\mathbb{E}\!\left[ F(\mathbf{w}^*)\right] \!+\!\frac{\varepsilon_u}{2\varpi_u^2} \!\right),
    \end{align}
    where $u\in\{\mathrm{I},\mathrm{II}\}$, the scaling factor  $\varpi_u$, and bias term $\varepsilon_u$ for Scheme $u$ are respectively given by
    \begin{align}
        \varpi_{\mathrm{I}}\!&=\!\left(\frac{(16\!-\!\pi^2)\xi^2}{\pi^2N}+1\right)L,\nonumber\\
        \varepsilon_\mathrm{I}\!&=\left(\!\frac{8(K\!+\!1)\!+\!\pi^2(N\!-\!1)}{\pi^2 NK}\chi^2+\!\Delta_{1,2}\!+\!\Delta_{1,3}\right)L,\nonumber \\
        \varpi_{\mathrm{II}}\!&=\!\left(\frac{\frac{8}{K^2}\sum_{k\in\mathcal{K}}\alpha_k^{2}\sum_{i\in\mathcal{K}}\alpha_i^{-2}+8-\pi^2}{\pi^2 N}\xi^2 +1\right)L,\nonumber\\
        \varepsilon_{\mathrm{II}}\!&=\!\!\left(\!\frac{\frac{8}{K}\!\sum_{k\in\!\mathcal{K}}\!\alpha_k^{2}\!\sum_{i\in\!\mathcal{K}}\alpha_i^{-2}\!\!+\!\!8\!\!+\!\!\pi^2\!(\!N\!-\!1)}{\pi^2 NK}\!\chi^2\!+\!\Delta_{2,2}\!+\!\Delta_{2,3}\!\right)\!\!L.
    \end{align}
\end{theorem}
\begin{proof}
    See Appendix \ref{appE}.$\hfill\blacksquare$
\end{proof}

According to the above theorem, we conclude the following.

\begin{remark}[Convergence Rate]
    With a given learning rate, the convergence rate in (\ref{th4-1}) of the proposed robust aggregation schemes is on the order of $\mathcal{O}(\frac{2\varpi_{u}}{\sqrt{T}})$. By invoking the Cauchy-Schwarz inequality, it is easily verified that $\varpi_\mathrm{I}\leq \varpi_{\mathrm{II}}$. Thus, it can be inferred that Scheme \Rmnum{1} consistently achieves faster convergence than Scheme \Rmnum{2}. This fast-convergent characteristic originates from a reduced computational error, as delineated in the MSE analysis. Moreover, the convergence rate of the suggested methodology is solely influenced by the data heterogeneity, $\xi$, and the number of RIS elements, $N$. It is impervious to interferers, underscoring significant advancement over prevailing approaches~\cite{22}.
\end{remark}

\begin{remark}[Limiting Performance]
     With increasing $N$, we note that the scaling factor  $\varpi_u\to L$ and bias factor $\varepsilon_u\to \frac{\chi^2}{K}$.
     This indicates that the AirFL system gradually approximates the ideal convergence without interference or additional noise, with its ultimate performance limited only by a SGD error,~$\chi^2$. Therefore, \emph{by deploying a large number of low-cost RIS reflecting elements and utilizing the proposed aggregation schemes, we pave the way to realize an asymptotically optimal FL algorithm over the air,  even with interference and low SNR conditions.}
\end{remark}

% \vspace{-15.pt}
\section{Numerical Results}
\begin{figure}[!t]
    \setlength{\abovecaptionskip}{0pt}
    \setlength{\belowcaptionskip}{0pt}
    \centering
    \includegraphics[width=3.5in]{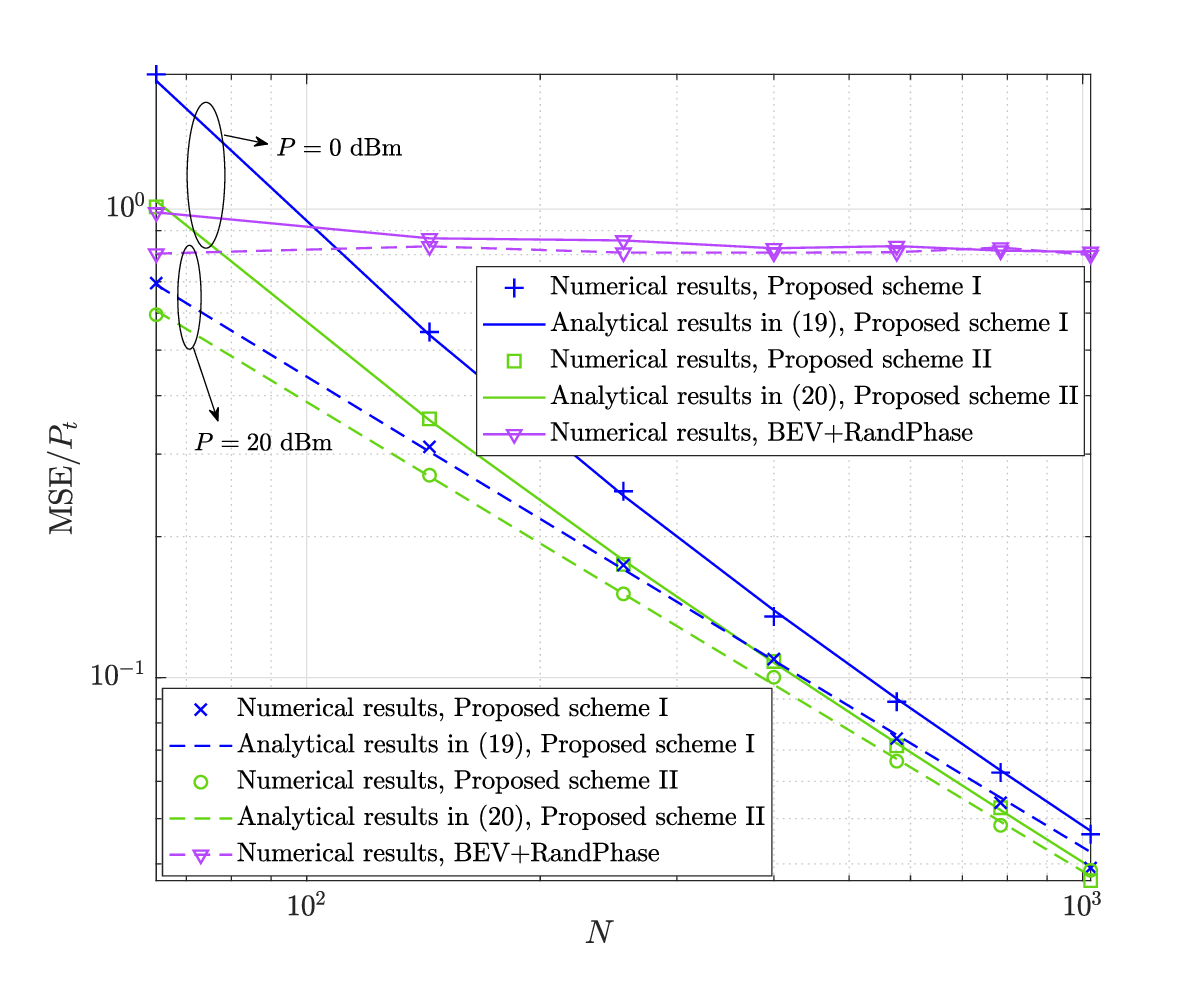}
    \caption{The MSE versus $N$ with $K=20$ and $M=10$.}
    % \vspace{-0.4cm}
    \label{fig2} \end{figure}
\begin{figure}[!t]
    \setlength{\abovecaptionskip}{0pt}
    \setlength{\belowcaptionskip}{0pt}
    \centering
    \includegraphics[width=3.5in]{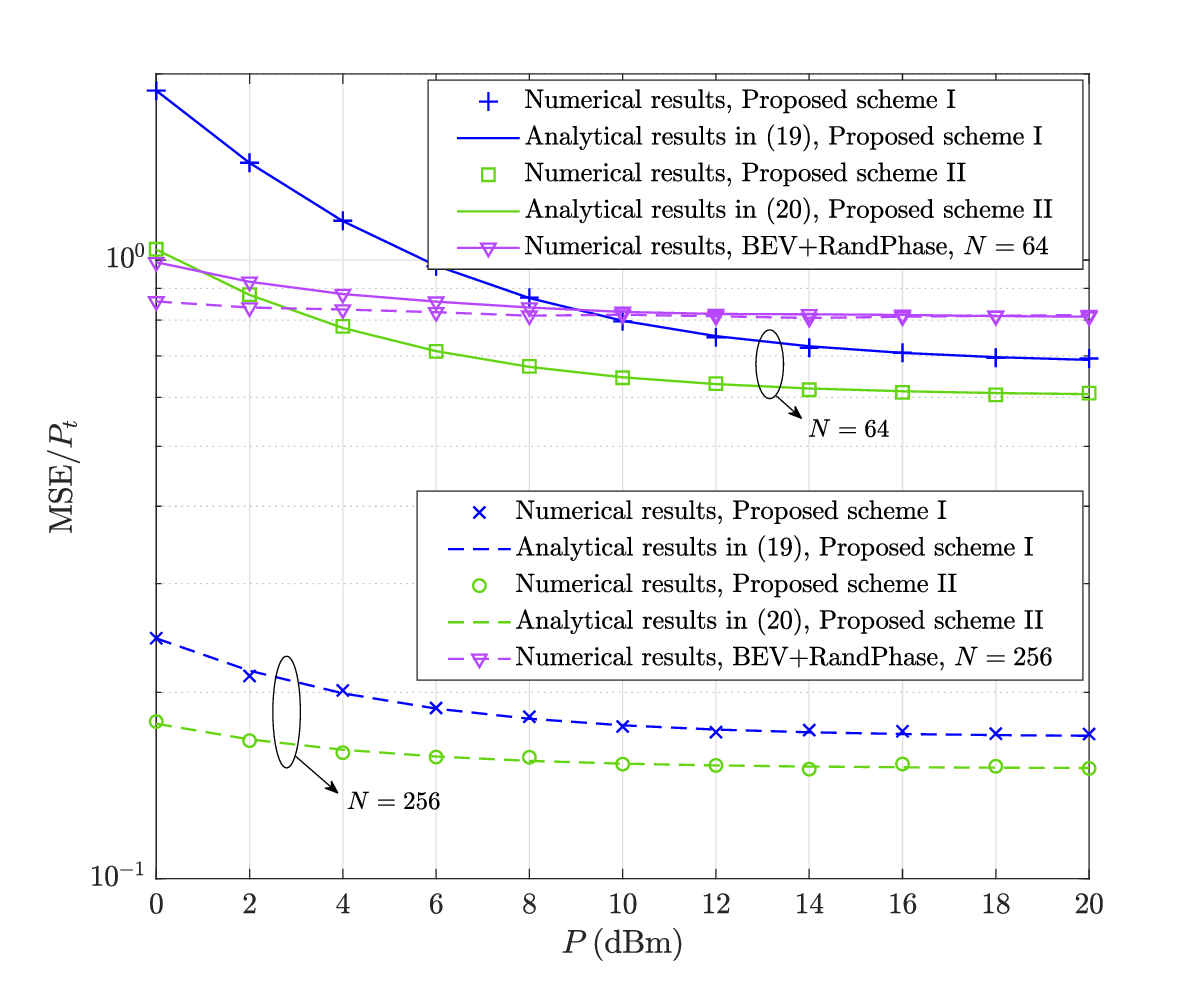}
    \caption{The MSE versus $P$ with $K=20$ and $M=10$.}
    % \vspace{-0.4cm}
    \label{fig3} \end{figure}
\begin{figure}[!t]
    \setlength{\abovecaptionskip}{0pt}
    \setlength{\belowcaptionskip}{0pt}
    \centering
    \includegraphics[width=3.5in]{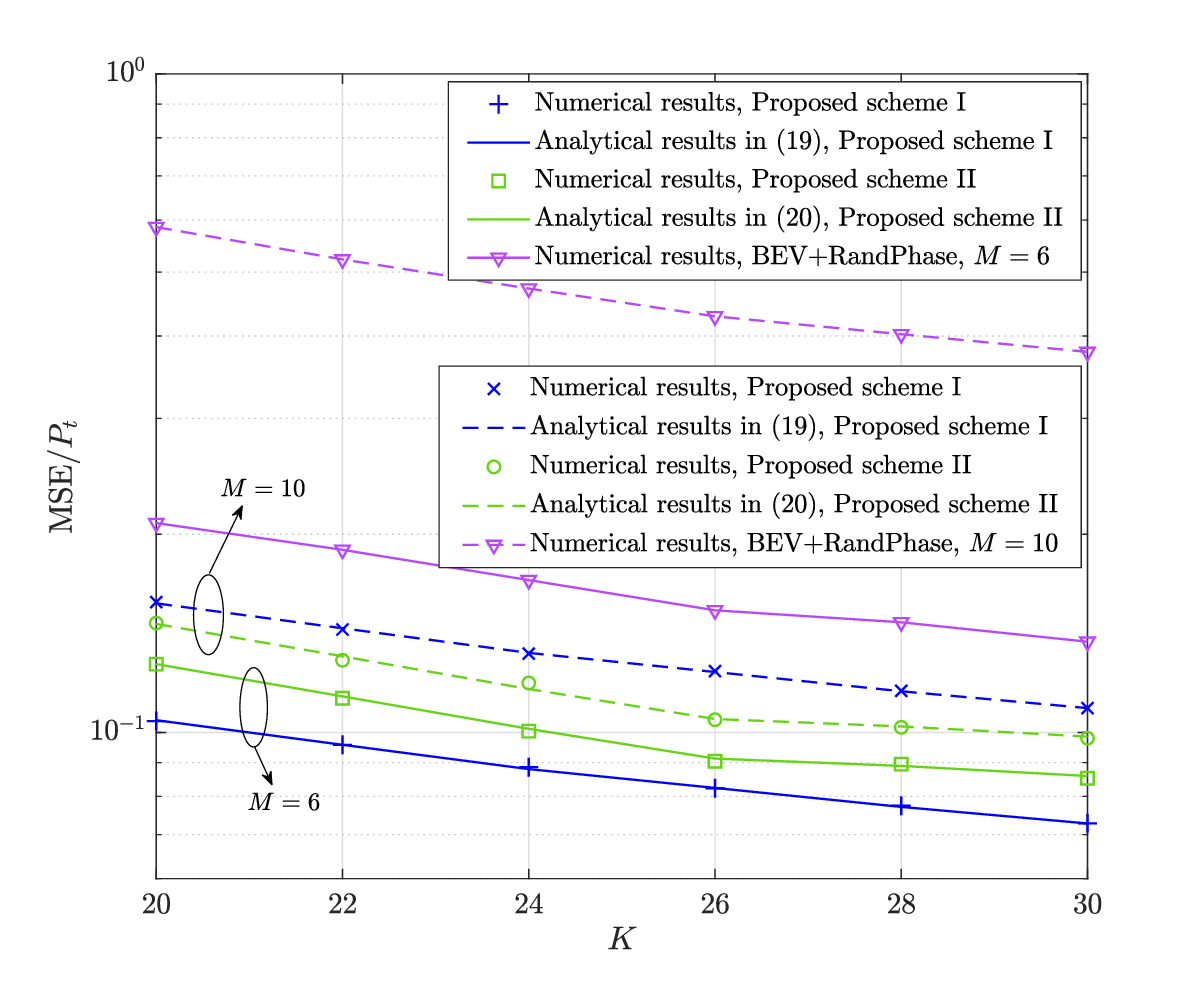}
    \caption{The MSE versus $K$ with $N=256$ and $P=0~{\rm dBm}$.}
    % \vspace{-0.4cm}
    \label{fig4} \end{figure}
\begin{figure}[!t]
    \setlength{\abovecaptionskip}{0pt}
    \setlength{\belowcaptionskip}{0pt}
    \centering
    \includegraphics[width=3.5in]{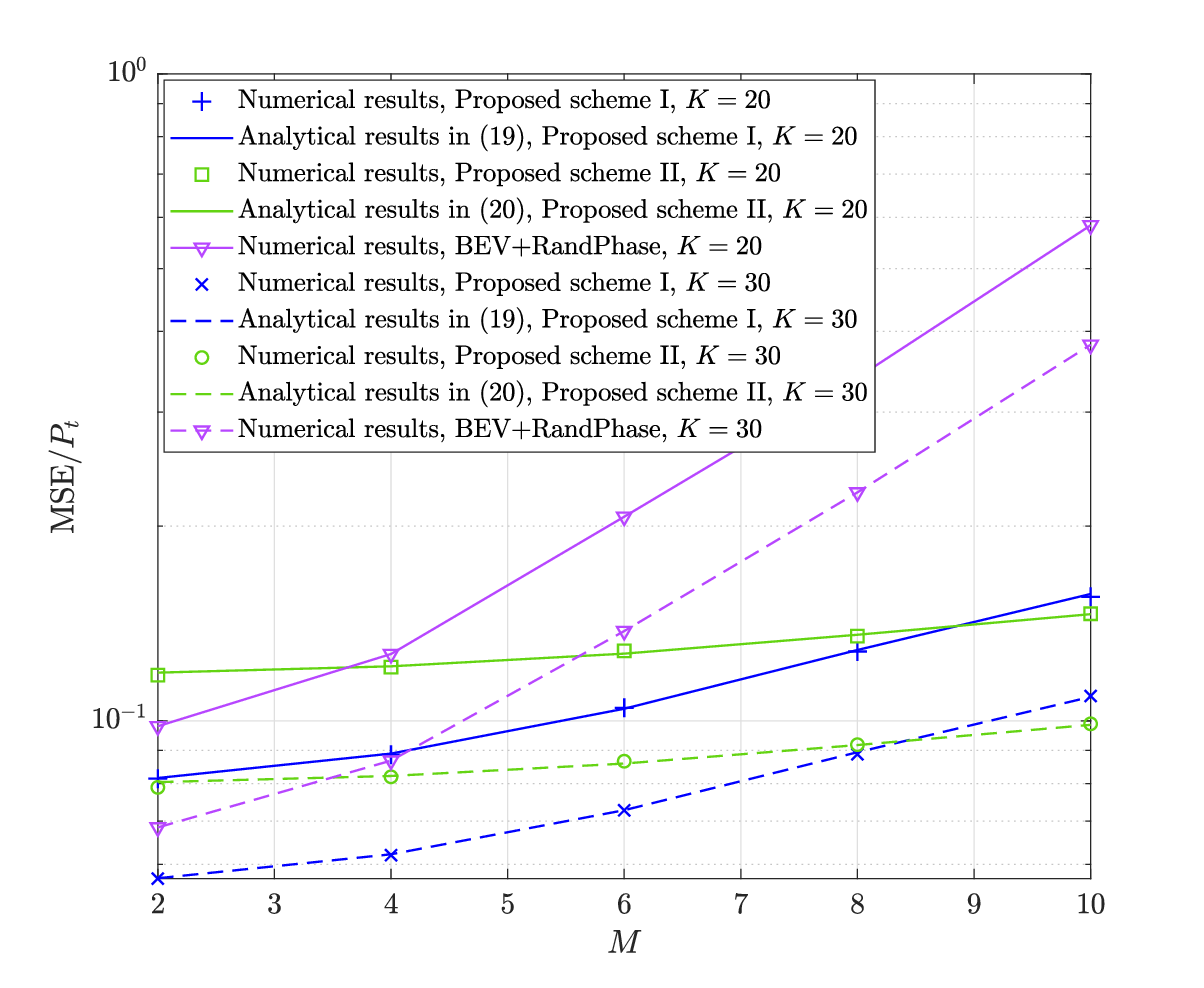}
    \caption{The MSE versus $M$ with $N=256$ and $P=0~{\rm dBm}$.}
    % \vspace{-0.4cm}
    \label{fig5} \end{figure}
In this section, numerical simulations are presented to validate the proposed schemes and analytical results.  We assume that the distance between the PS and RIS is $200~{\rm m}$, and all the devices are uniformly distributed within a disk of radius $300~{\rm m}$ centered at the RIS. The path loss exponent for all the links is $2.2$. The maximum transmit power of each device is the same, denoted as $P$. Unless otherwise specified, the other parameters are set as the number of target devices $K=20$, the number of interference devices $M=10$, the bandwidth $B=1$ MHz, the noise power spectral density $N_0=-140$ dBm/Hz, and the maximum transmit power $P=0$ dBm.

To evaluate the learning performance, we perform the FL tasks of image classification on the two popular datasets, i.e., MNIST and CIFAR-10. For the MNIST dataset, a multi-layer perceptron (MLP) with $D=23860$ parameters is trained via the AirFL. Regarding CIFAR-10, we adopt a convolutional neural network (CNN) with $D=62000$ parameters. It is noteworthy that all local datasets are non-IID, comprising at most two categories of labels. For interference devices, we assume that malicious zero-gradient attack are performed \cite{19}. The learning parameters are set as the batch size $b_k=50$ and the learning rate $\eta_t=0.005$.

For performance comparison, we mainly consider the following baseline schemes.
\begin{itemize}
    \item BEV+RandPhase: the target devices perform the BEV power control strategy in \cite{22} to combat interference and the RIS phases are randomly selected.
    \item BEV+RR: the target devices perform the BEV strategy in \cite{22} and RIS sequentially aligns to each target device's channel, similar to the Round Robin scheduling in \cite{RR}.
    \item BEV-RO: The random orthogonalization scheme with a multi-antenna receiver at the PS in \cite{36} is adopted for model aggregation and interference suppression. For the sake of fairness, compared to the RIS-assisted link, the direct channel in this scenario has a shorter distance but a larger path loss exponent, which is set as 3.5.
    \item BEV-minMSE: Each device adopts a BEV power control strategy and we jointly optimize the RIS phase shifts and the denoising factor at the receiver, aiming to minimize the MSE of gradient estimation, similar to \cite{minMSE}.
\end{itemize}

% \vspace{-13.pt}
\subsection{MSE Performance}
To provide a relative measure of error that can be compared across different datasets and models, we normalize the MSE by dividing (\ref{MSE}) by the power of the global gradient $\mathbf{g}_t$, which is defined as $P_t\!=\!\mathbb{E}\!\left[\Vert\mathbf{g}_t\Vert^2\right]$. This normalization process does not impact the analytical results in terms of $N$ and~$P$. We compare the results from Monte-Carlo simulations with (\ref{MSE1}) and (\ref{MSE2}) here. Fig.~\ref{fig2} depicts the normalized MSE versus the number of RIS elements, $N$. We see that both the analytical results match well with the numerical results. Furthermore, as predicted in \emph{Remark~1}, the MSE of our proposed schemes decreases linearly with $N$ on a log-log scale. This phenomenon becomes more obvious as $P$ increases, owing to the diminishing impact of noise error. Contrarily, the baseline scheme utilizing random RIS phase shifts fails to obtain any effective performance enhancements as $N$ increases, which demonstrates the importance of RIS phase shift configurations.

% \vspace{-12.pt}    
Fig.~\ref{fig3} shows the MSE versus the maximum transmit power $P$ for different values of $N$. It is evident that, compared to the marginal gains from increasing $P$, the MSE experiences more significant improvements as $N$ increases. Moreover, a performance ceiling is observed for the MSE as $P$ grows, which occurs because the MSE converges to a positive constant rather than approaching zero, as discussed in \emph{Remark~2}. Consequently, increasing $P$ proves to be less effective than increasing the number of RIS reflecting elements $N$ in reducing MSE.

Fig.~\ref{fig4} and Fig.~\ref{fig5} illustrate the MSE as a function of the number of target devices, $K$, and the number of interference devices, $M$, respectively. It is clearly shown that increasing $K$ and decreasing $M$ both improve the MSE performance. We can further observe that, for relatively large values of $M$, Scheme~\Rmnum{2} outperforms Scheme~\Rmnum{1} in terms of the MSE. Conversely, when $M$ is relatively small, Scheme~\Rmnum{1} exhibits superior MSE. This implies that Scheme~\Rmnum{1} achieves more efficient model aggregation, while Scheme~\Rmnum{2} demonstrates better performance in mitigating interference, validating the conclusions presented in \emph{Observations~1} and \emph{2}. In addition, it can be seen that compared with the baselines, our proposed schemes exhibit better advantages when interference is severe.

\begin{figure}[!t]
    \setlength{\abovecaptionskip}{0pt}
    \setlength{\belowcaptionskip}{0pt}
    \centering
    \includegraphics[width=3.5in]{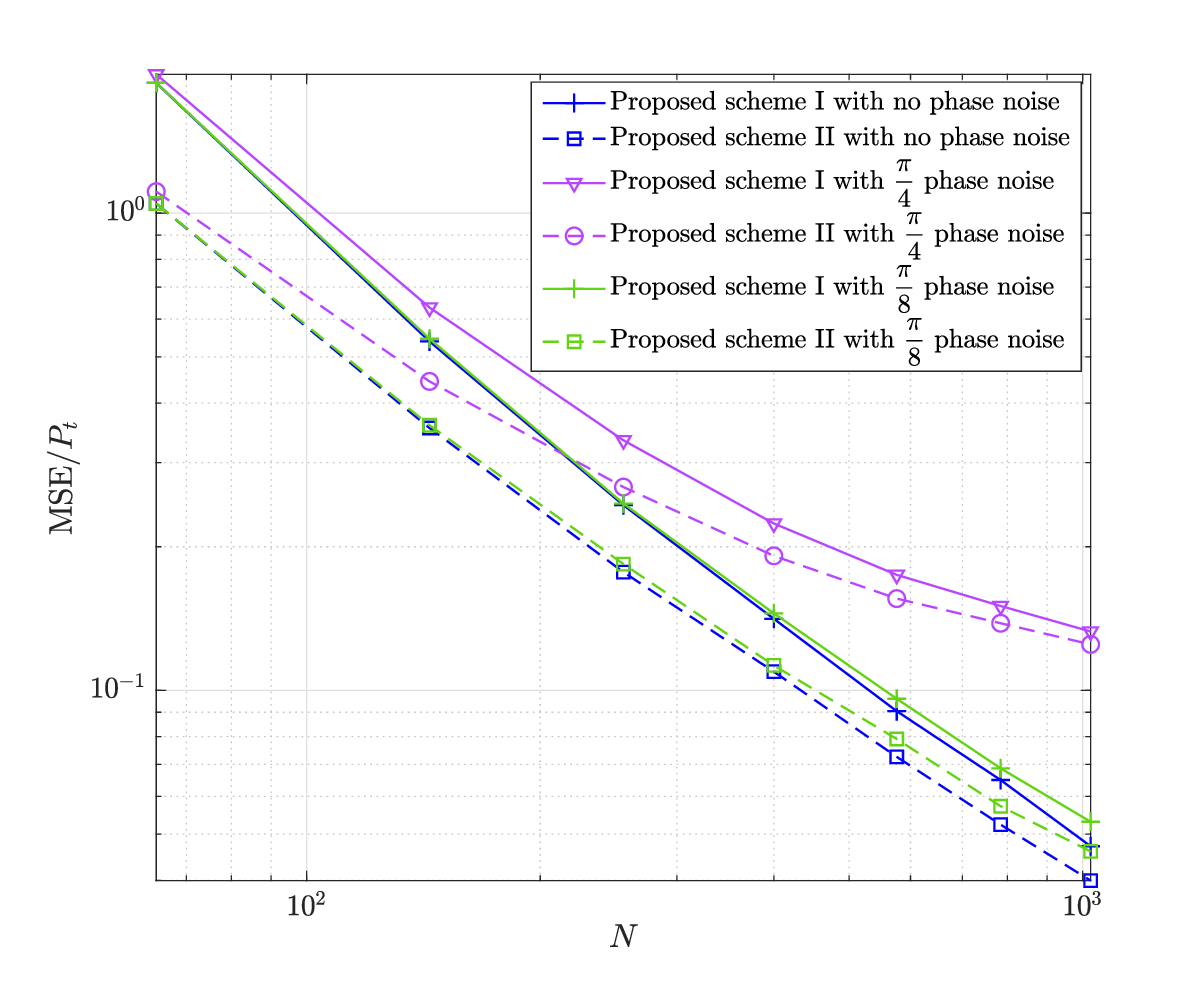}
    \caption{The impact of RIS phase noise on the MSE with $K=20$, $M=10$, and $P=0~{\rm dBm}$.}
    % \vspace{-0.4cm}
    \label{fig1-2} 
\end{figure}

Fig. \ref{fig1-2} illustrates the impact of RIS phase noise on the MSE when $K\!=\!20$, $M\!=\!10$ and $P\!=\!0~{\rm dBm}$. We~observe that for phase noise with a deviation of $\frac{\pi}{8}$, its impact on the MSE is negligible. This observation is consistent with the result shown in \cite{wshi}, confirming that RIS-assisted communications with 3-bit discrete phase shifts (capable of achieving up to~$\frac{\pi}{8}$ phase noise) asymptotically achieve the ideal performance of continuous phase shifts. This demonstrates the applicability of our proposed schemes under practical implementations.

% \vspace{-0.3cm}
\subsection{Convergence Performance}
\begin{figure}[!t]
    \centering
    \begin{minipage}[t]{1\linewidth}
        \centering
        \includegraphics[width=1\linewidth]{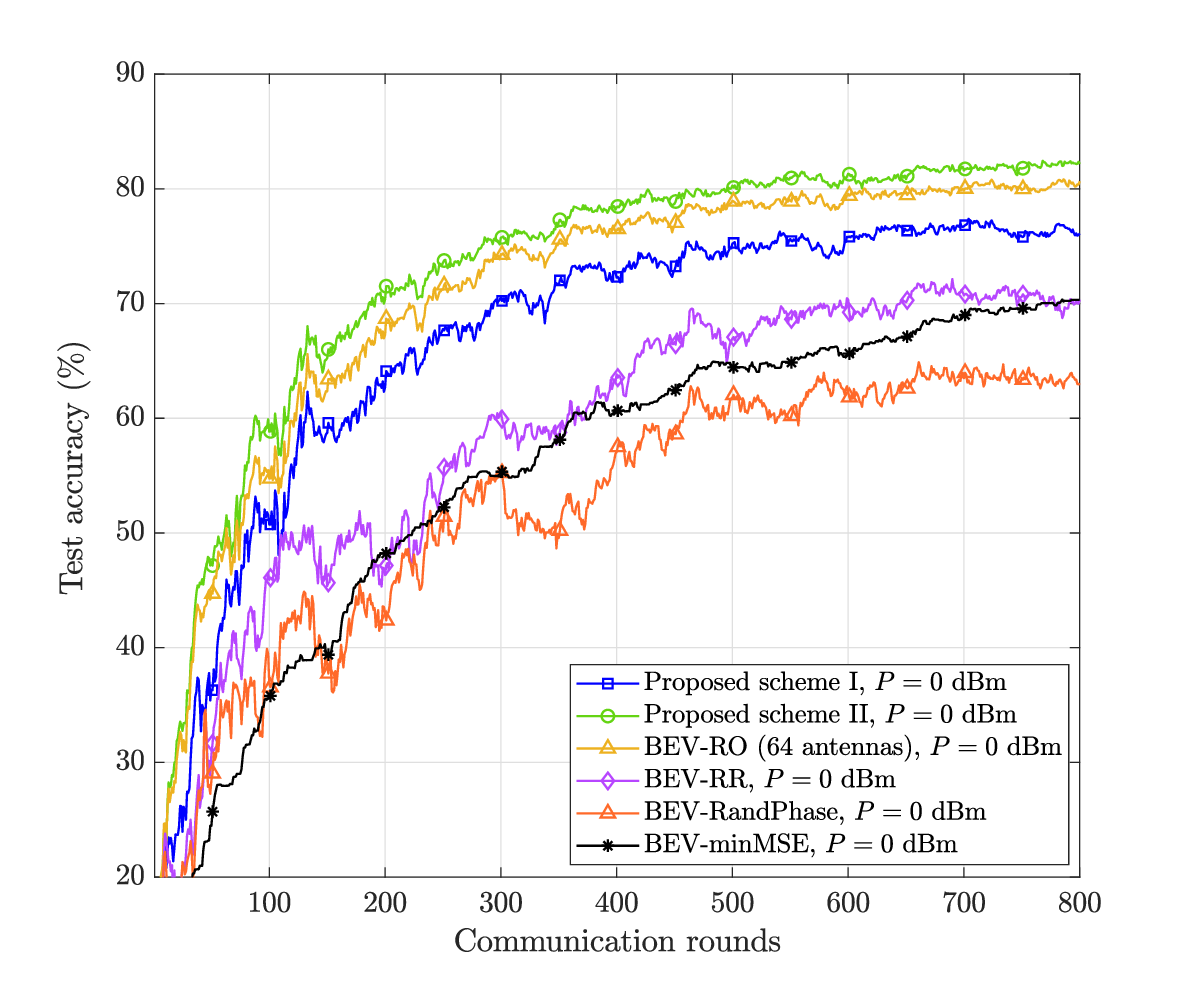}
    \end{minipage}
    \begin{minipage}[t]{1\linewidth}
        \centering
        \includegraphics[width=1\linewidth]{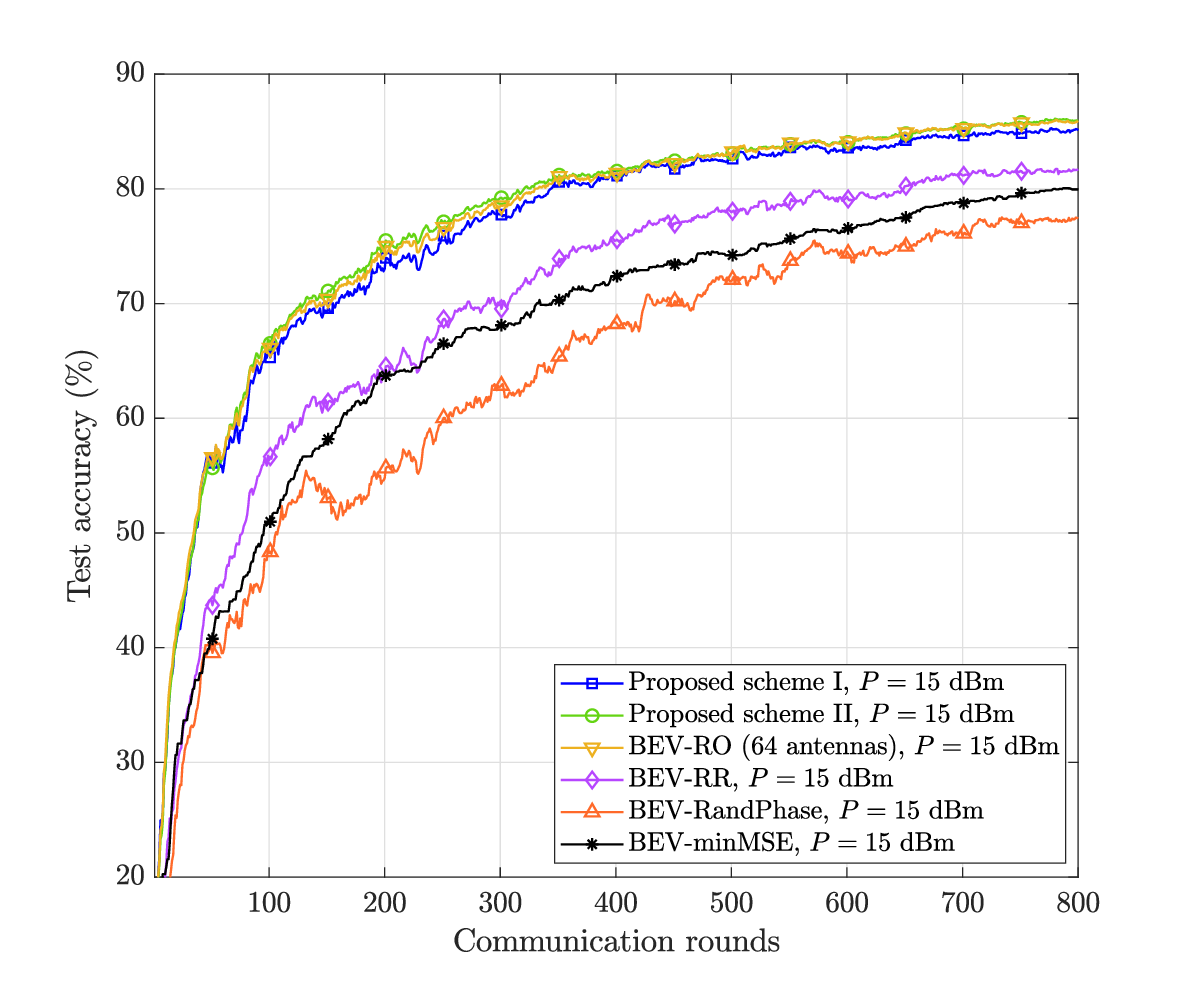}
    \end{minipage}
 \caption{Test accuracy versus communication rounds on MNIST datasets.}
 % \vspace{-0.4cm}
 \label{fig6} 
\end{figure}

\begin{figure}[!t]
    \centering
    \begin{minipage}[t]{1\linewidth}
        \centering
        \includegraphics[width=1\linewidth]{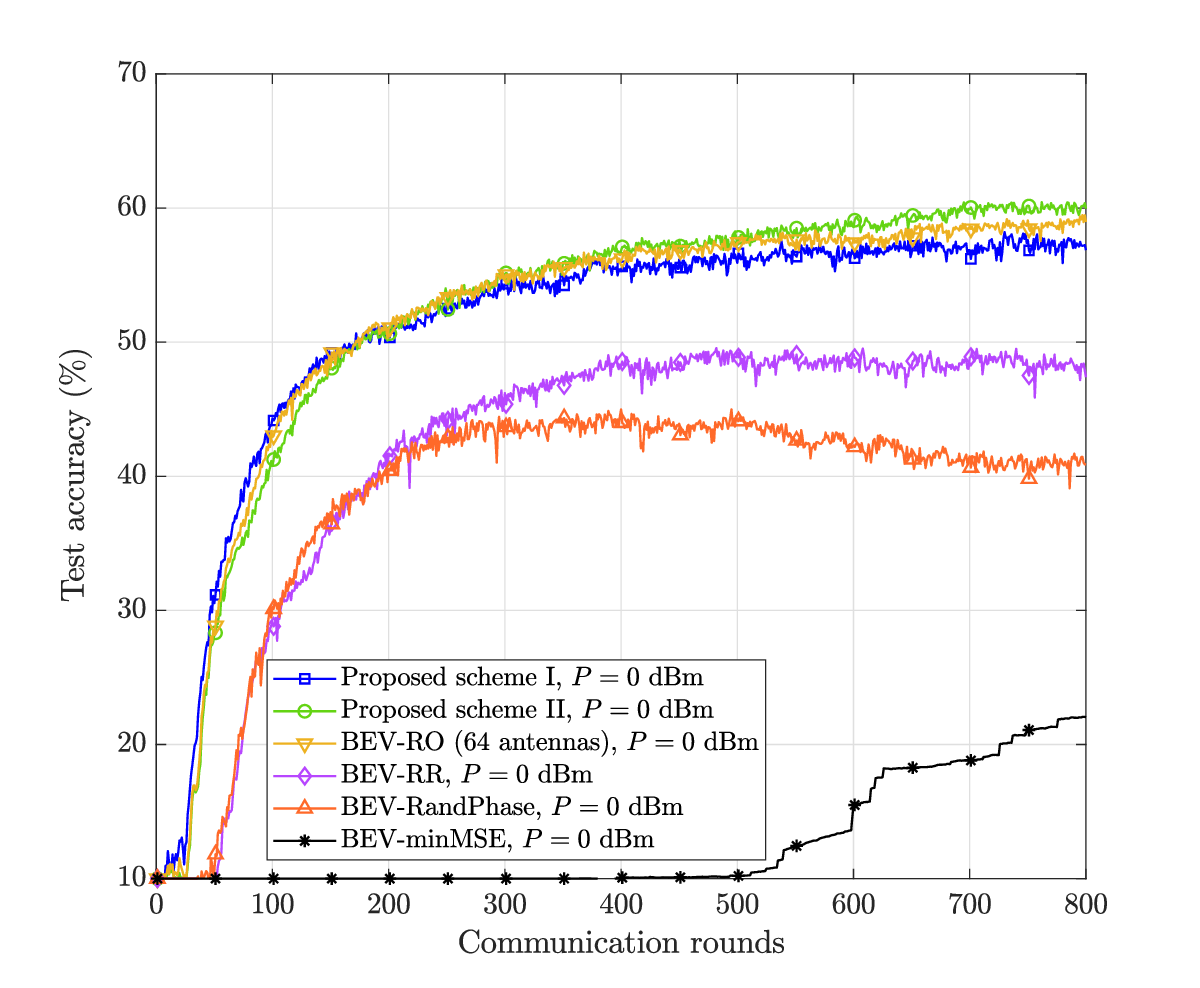}
    \end{minipage}
    \begin{minipage}[t]{1\linewidth}
        \centering
        \includegraphics[width=1\linewidth]{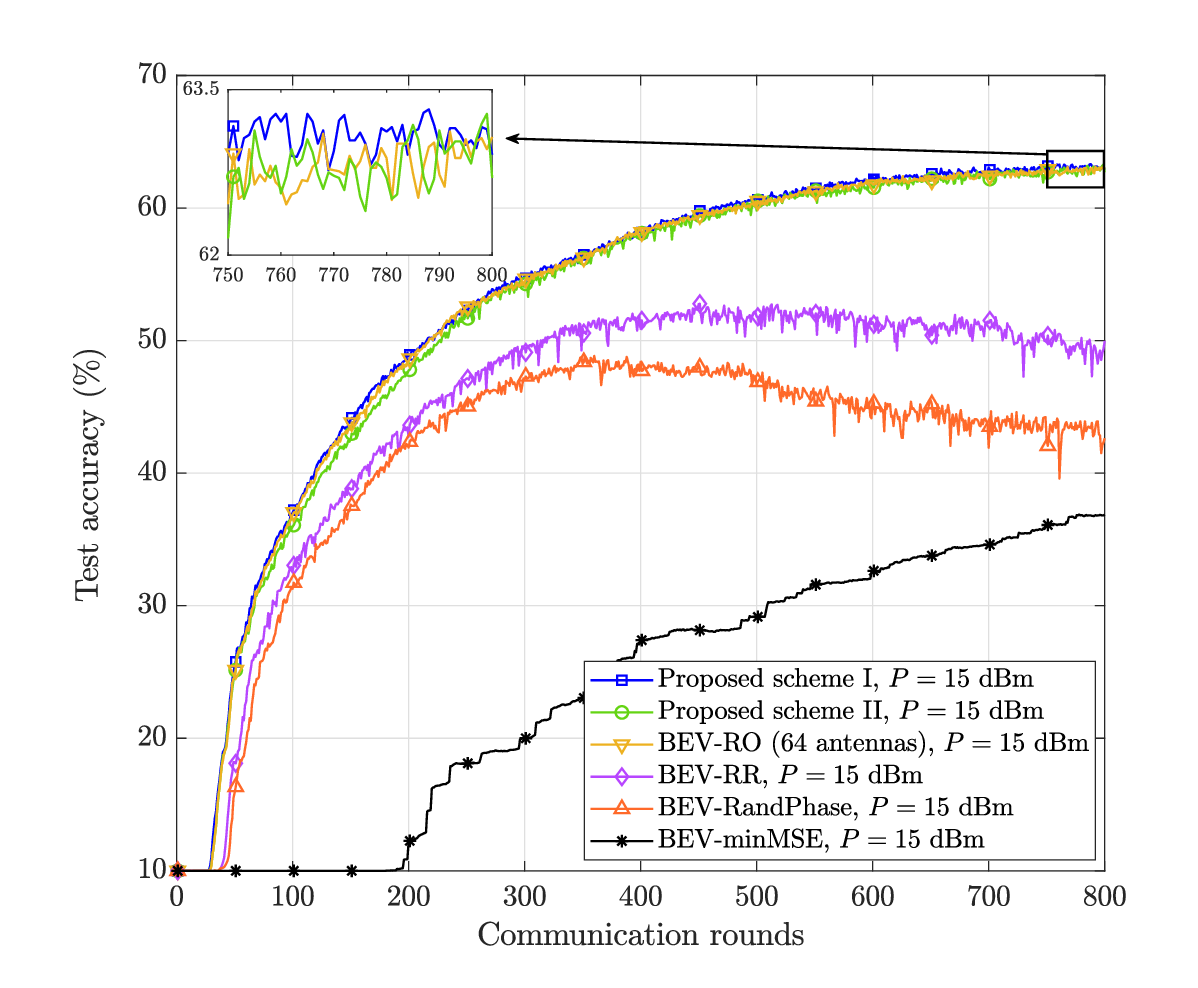}
    \end{minipage}
 \caption{Test accuracy versus communication rounds on CIFAR-10 datasets.}
 % \vspace{-0.4cm}
 \label{fig7} 
\end{figure}

Fig. \ref{fig6} illustrates the testing accuracy on the MNIST datasets for transmission powers of 0 and 15 dBm, respectively. Firstly, we observe that the proposed schemes perform comparably to the MIMO-based RO scheme. This can be attributed to the deployment of RIS, which significantly reduces the path loss exponent relative to direct links. Furthermore, both the proposed and RO schemes capitalize on the channel hardening and favorable propagation effects of large-scale antenna arrays, without requiring additional RF chains for signal processing. Hence, the primary function of large-scale arrays at the receiver to achieve diversity gain is effectively realized by RIS. However, to compensate for the double path fading effect introduced by RIS, a higher number of RIS elements are required compared to MIMO. Our results demonstrate that deploying 256 low-cost RIS elements surpasses the performance of a 64-antenna MIMO setup, highlighting the proposed RIS-based scheme's cost-efficiency advantage over MIMO.

Additionally, Fig. \ref{fig6} shows that the proposed schemes substantially outperform other baseline schemes aside from the MIMO-based RO scheme, confirming their effectiveness. Notably, at a low SNR regime, Scheme \Rmnum{2} exhibits a more pronounced advantage over Scheme \Rmnum{1}. This is because the FL convergence with low SNR level is primarily constrained by noise, underscoring Scheme \Rmnum{2}'s exceptional capability in suppressing interference and noise. As the SNR increases, the initial advantages of Scheme \Rmnum{2} gradually diminish, and both schemes tend to converge to a comparable performance level.

Fig. \ref{fig7} presents the performance of the proposed schemes on the CIFAR-10 datasets. Consistent with the observations on the MNIST datasets, our schemes exhibit notable superiority over baseline schemes, except for the MIMO-based RO scheme. Moreover, considering the inherently more intricate nature of the classification task associated with the CIFAR-10 datasets, it becomes evident that existing baseline schemes may encounter challenges in converging under conditions of strong interference. Aligned with our theoretical analysis, Scheme~\Rmnum{1} demonstrates faster convergence than Scheme~\Rmnum{2}, though this benefit comes at the expense of performance loss at low SNRs. Therefore, each scheme offers distinct advantages, with selection depending on the specific channel conditions and requirements.

Traditional baseline schemes, which focus on minimizing the MSE, perform poorly in interference-dominated scenarios. This is because, to counteract strong interference and noise, the receiver typically increases the denoising factor, $\lambda$, to maintain a low MSE. However, when interference becomes overwhelming, $\lambda$ must be raised significantly to minimize MSE, causing the estimated gradient to approach zero. As a result, the optimal MSE converges to $\mathbb{E}\left[\Vert \mathbf{g}_t \Vert^2 \right]$. In summary, the receiver, aiming to avoid excessive MSE, conservatively estimates a smaller gradient, slowing the gradient descent process. This also explains why such schemes exhibit more severe performance degradation, particularly at low SNR. In contrast, the proposed scheme prioritizes the unbiasedness of gradient estimation, ensuring superior convergence performance than traditional MSE-minimizing approaches.

% \vspace{-5.pt}
\section{Conclusion}
In this paper, we have proposed a novel concept of phase-manipulated favorable propagation and channel hardening via RIS to achieve robust gradient aggregation in the AirFL system with external interference. Specifically, two transmission schemes with different power allocation and RIS phase shift settings have been proposed to guarantee unbiased gradient estimation. Then, both MSE and FL convergence analyses were conducted to affirm the anti-interference capability of the proposed schemes. The obtained results quantify the impact of key parameters on the MSE and FL convergence and provide insightful guidelines for system design. 
Several simulations were provided to demonstrate the analytical results and validate the superior performance of the proposed schemes over existing baselines. 

There are several interesting research directions for future work. One direction is to integrate the proposed approach with advanced FL algorithms, such as FedProx \cite{FedProx} and personalized FL \cite{pAirFL}, to better address data heterogeneity in~scenarios with non-IID local datasets. Additionally, exploring novel types of RIS beyond the passive RIS studied in this paper, such as active RIS \cite{aRIS}, could offer further benefits by mitigating double path fading effect.

% \vspace{-5.pt}
\begin{appendices}
\section{Preliminary Lemmas}\label{appA}
Some useful lemmas are formally introduced as follows, which will be used in the later derivations.

\begin{lemma}\label{lemma1}
    If $x$ and $y$ are correlated Rayleigh RVs with mean $\frac{\sqrt\pi}{2}$ and variance $\frac{4-\pi}{4}$, then we get $\mathbb{E}\left[\frac{x^2}{y}\right]=\frac{\sqrt\pi}{2}(2-\rho)$, where the correlation coefficient $\rho=\mathbb{E}\left[x^2y^2\right]-1$.
\end{lemma}

\begin{proof} 

By applying the joint probability density function (PDF) given in \cite[Eq.~(1)]{bivariate}, we obtain
\begin{align}\label{lemma1-E2}
\mathbb{E}\left[\frac{x^2}{y}\right]&=\int_{0}^{+\infty}\int_{0}^{+\infty}{\frac{4x^3{\rm e}^{-\frac{x^2+y^2}{1-\rho}}}{1-\rho}I_0\left\{\frac{2\sqrt\rho x y}{1-\rho}\right\}}{\rm{d}}x{\rm{d}}y\nonumber\\
&\mathop=^{\left(\rm a\right)}\int_{0}^{+\infty}{\frac{2\sqrt\pi}{\sqrt{1-\rho}}x^3{\rm e}^\frac{\left(\rho-2\right)x^2}{2\left(1-\rho\right)}I_0\left\{\frac{\rho x^2}{2\left(1-\rho\right)}\right\}}{\rm{d}}x\nonumber\\
&\mathop=^{\left(\rm b\right)}\int_{0}^{+\infty}{\frac{\sqrt\pi}{\sqrt{1-\rho}}t{\rm e}^\frac{\left(\rho-2\right)t}{2\left(1-\rho\right)}I_0\left\{\frac{\rho t}{2\left(1-\rho\right)}\right\}}{\rm{d}}t\nonumber\\
&\mathop=^{\left(\rm c\right)}\frac{\sqrt\pi}{2}(2-\rho),
\end{align}
where $I_\nu\left\{\cdot\right\}$ denotes the $\nu$-th-order modified Bessel function of the first kind \cite[Eq.~(8.406)]{table}, the coefficient $\rho\!=\!\frac{\mathbb{E}\left[x^2y^2\right]-\mathbb{E}\left[x^2\right]\mathbb{E}\left[y^2\right]}{\sqrt{\mathbb{V}\left[x^2\right]\mathbb{V}\left[y^2\right]}}$, $({\rm a})$ is obtained from \cite[Eq.~(6.618.4)]{table}, $({\rm b})$ follows by letting $t=x^2$, and $({\rm c})$ is calculated by using \cite[Eq.~(6.623.2)]{table}. By substituting $\mathbb{E}\left[x^2\right]=\mathbb{E}\left[y^2\right]=1$ and $\mathbb{V}\left[x^2\right]=\mathbb{V}\left[y^2\right]=1$, we complete the proof. $\hfill\blacksquare$
\end{proof}

\begin{lemma}\label{lemma4}
    If $x$ and $y$ are independent exponential RVs, $x\!\sim\!{\rm Exp}\!\left(\lambda_1\right)$, $y\!\sim\!{\rm Exp}\!\left(\lambda_2\right)$, then $z\!=\!\min\!\left\{x,y\right\}\!\sim\!{\rm Exp}\!\left(\lambda_1+\lambda_2\right)$.
\end{lemma}
\begin{lemma}\label{lemma2}
    If $x$ and $y$ are independent uniformly distributed RVs, $x\!\sim\!{\mathcal{U}}\!\left(0,2\pi\right)$, $y\!\sim\!{\mathcal{U}}\!\left(0,2\pi\right)$, then the PDF of $z\!=\!x\!+\!y$ is 
\begin{align}
f_z\left(z\right)=
\begin{cases}
\frac{z}{4\pi^2},~~~0\le z<2\pi,\\
\frac{4\pi-z}{4\pi^2},~2\pi\le z\le4\pi.
\end{cases}
\end{align}
\end{lemma}

\begin{lemma}\label{lemma3}
    If $x$ and $y$ are independent uniformly distributed RVs, $x\!\sim\!{\mathcal{U}}\!\left(0,2\pi\right)$, $y\!\sim\!{\mathcal{U}}\!\left(0,2\pi\right)$, then the PDF of $z\!=\!x\!-\!y$ is 
\begin{align}
f_z\left(z\right)=
\begin{cases}
\frac{2\pi+z}{4\pi^2},~-2\pi\le z<0,\\
\frac{2\pi-z}{4\pi^2},~~~~0\le z\le2\pi.
\end{cases}
\end{align}
\end{lemma}

\section{Proof of Theorem \ref{theorem1}}\label{appB}
By substituting the RIS phase shifts in (\ref{eq10}), the mean of $ u_k\triangleq \Re \left\{ \mathbf{h}_p^H \mathbf{\Theta}\mathbf{h}_{r,k}\right \}$ is calculated as
\begin{align}
\mathbb{E}[u_k]&=\mathbb{E}\!\!\left[\!\Re\!\left\{\!\sum_{n=1}^{N}\!{\left|h_{p,n}^\ast\right|\!h_{r,k,n}\frac{\sum_{i=1}^{K}\!{w_ih_{r,i,n}^\ast}}{\left|\sum_{i=1}^{K}\!{w_ih_{r,i,n}^\ast}\right|}}\!\right\}\!\right]\nonumber\\
&=\mathbb{E}\!\!\left[\!\sum_{n=1}^{N}\!\left|h_{p,n}^\ast\right|\!\Re\!\left\{\!h_{r,k,n}\frac{\sum_{i=1}^{K}\!{w_ih_{r,i,n}^\ast}}{\left|\sum_{i=1}^{K}\!{w_ih_{r,i,n}^\ast}\right|}\!\right\}\!\right]\nonumber\\
&=\frac{\sqrt\pi}{2}\sum_{n=1}^{N}\mathbb{E}\!\left[\!\Re\!\left\{\!h_{r,k,n}\frac{\sum_{i=1}^{K}{w_ih_{r,i,n}^\ast}}{\left|\sum_{i=1}^{K}{w_ih_{r,i,n}^\ast}\right|}\right\}\right],
\label{appb1}
\end{align}
where the last step comes from the independence of $h_{p,n}$ and $h_{r,i,n}$, and using $\mathbb{E}\left[|h_{p,n}^\ast|\right]=\frac{\sqrt{\pi}}{2}$ since $h_{p,n}\sim \mathcal{CN}(0,1)$ \cite{wshi}.
% \textcolor{red}{Also, we omit subscripts $r, n$ for notational simplicity.} 

Further, by defining $a\!=\!w_kh_{r,k,n}\!\sim\!\mathcal{CN}\!\left(0,w_k^2\right)$ and $b\!=\!\sum_{i\neq k}{w_ih_{r,i,n}}\!\sim\!\mathcal{CN}\!\left(0,\sum_{i\neq k}w_i^2\right)$, we obtain
\begin{align}\label{eq13}
&\mathbb{E}\!\left[\Re\!\left\{h_{r,k,n}\frac{\sum_{i=1}^{K}{w_ih_{r,i,n}^\ast}}{\left|\sum_{i=1}^{K}{w_ih_{r,i,n}^\ast}\right|}\right\}\right]\nonumber \\
&=\frac{1}{w_k}\mathbb{E}\!\left[\Re\left\{a\frac{a^\ast+b^\ast}{\left|a^\ast+b^\ast \right|}\right\}\right]=\frac{1}{w_k}\mathbb{E}\!\left[\frac{2\left|a\right|^2+2\Re\left\{a b^\ast\right\}}{2\left|a+b\right|}\right]\nonumber \\
&=\frac{1}{w_k}\mathbb{E}\!\left[\frac{\left|a+b\right|^2+\left|a\right|^2-\left|b\right|^2}{2\left|a+b\right|}\right]\nonumber\\
&=\frac{1}{2w_k}\left\{\mathbb{E}\!\left[\left|a+b\right|\right]+\mathbb{E}\!\left[\frac{\left|a\right|^2}{\left|a+b\right|}\right]-\mathbb{E}\!\left[\frac{\left|b\right|^2}{\left|a+b\right|}\right]\right\}.
\end{align}
Noting that $\vert a+b \vert$, $\vert a\vert$, and $\vert b\vert$ are Rayleigh RVs, we calculate each term in (\ref{eq13}) as
\begin{align}\label{eq14}
\mathbb{E}\!\left[\left|a\!+\!b\right|\right]&\!=\!\frac{\sqrt\pi}{2}\sqrt{\sum_{i=1}^{K}w_i^2},\nonumber \\
\mathbb{E}\!\left[\frac{\left|a\right|^2}{\left|a\!+\!b\right|}\right]&\!=\!\mathbb{E}\!\left[\frac{\left|w_kh_k\right|^2}{\left|\sum_{i=1}^{K}{w_ih_i}\right|}\right] \!\overset{\mathrm{(a)}}{=}\!\frac{\sqrt{\pi}w_k^2}{2\sqrt{\sum_{i=1}^{K}w_i^2}}\left(2\!-\!\rho_1\right),\nonumber \\
\mathbb{E}\!\left[\frac{\left|b\right|^2}{\left|a\!+\!b\right|}\right]&\!=\!\mathbb{E}\!\left[\frac{\left|\sum_{i\neq k}{w_ih_i}\right|^2}{\left|\sum_{i=1}^{K}{w_ih_i}\right|}\right] \!\overset{\mathrm{(b)}}{=}\!\frac{\sqrt{\pi}\sum_{i\neq k}w_i^2}{2\sqrt{\sum_{i=1}^{K}w_i^2}}\left(2\!-\!\rho_2\right),
\end{align}
where $({\rm a})$ and $({\rm b})$ apply \emph{Lemma~1} in Appendix~\ref{appA}, and the correlation coefficients $\rho_1$ and $\rho_2$ are, respectively, given by
\begin{align}\label{Theorem2-E6}
\rho_1&=\mathbb{E}\left[\left|h_k\right|^2\frac{1}{\sum_{i=1}^{K}w_i^2}\left|\sum_{i=1}^{K}{w_ih_i}\right|^2\right]-1\nonumber\\
&=\frac{1}{\sum_{i=1}^{K}w_i^2}\mathbb{E}\left[w_k^2\left|h_k\right|^4+\left|h_k\right|^2\sum_{i\neq k}{w_i^2\left|h_i\right|^2}\right]-1\nonumber\\
&=\frac{2w_k^2+\sum_{i\neq k}w_i^2}{\sum_{i=1}^{K}w_i^2}-1=\frac{w_k^2}{\sum_{i=1}^{K}w_i^2},
\end{align}
and
\begin{align}\label{Theorem2-E7}
\!\!\rho_2&\!=\!\mathbb{E}\left[\frac{1}{\sum_{i\neq k}w_i^2}\left|\sum_{i\neq k}{w_ih_i}\right|^2\frac{1}{\sum_{i=1}^{K}w_i^2}\left|\sum_{i=1}^{K}{w_ih_i}\right|^2\right]-1\nonumber\\
&\!=\!\frac{1}{\sum_{i\!\neq\! k}w_i^2}\frac{1}{\sum_{i\!=\!1}^{K}w_i^2}\mathbb{E}\!\left[\left|\sum_{i\!\neq \!k}{w_ih_i}\right|^4\!\!\!+\!w_k^2\!\left|h_k\right|^2\!\left|\sum_{i\!\neq\! k}{w_ih_i}\right|^2\right]\!\!-\!1\nonumber\\
&\!=\!\frac{2\left(\sum_{i\neq k}w_i^2\right)^2+w_k^2\sum_{i\neq k}w_i^2}{\sum_{i\neq k}w_i^2\sum_{i=1}^{K}w_i^2}-1=\frac{\sum_{i\neq k}w_i^2}{\sum_{i=1}^{K}w_i^2}.
\end{align}
Then, by substituting (\ref{eq14})--(\ref{Theorem2-E7}) into (\ref{eq13}), we have 
\begin{align}\label{Theorem2-E8-add1}
\mathbb{E}\!\left[\Re\!\left\{h_{r,k,n}\frac{\sum_{i=1}^{K}{w_ih_{r,i,n}^\ast}}{\left|\sum_{i=1}^{K}{w_ih_{r,i,n}^\ast}\right|}\right\}\right]=\frac{\sqrt\pi}{2}\frac{w_k}{\sqrt{\sum_{i=1}^{K}w_i^2}},
\end{align} 
which is same for all $n$. Therefore, we obtain
\begin{align}\label{Theorem2-E8}
\mathbb{E}[u_k]=\frac{\pi Nw_k}{4\sqrt{\sum_{i=1}^{K}w_i^2}},
\end{align} 
and finally arrive at the result of $\mathbb{E}\left[{\ell_k}\right]$ in (\ref{th1_1}).

In addition, the mean of $u_m\triangleq \Re\!\left\{\mathbf{h}_p^H\mathbf{\Theta}\mathbf{h}_{r,m}\!\right\}$, is given as
\begin{align}
\!\mathbb{E}[u_m] \!=\!\mathbb{E}\!\left[\!\Re\!\left\{\sum_{n=1}^{N}{\left|h_{p,n}^\ast\right|h_{r,m,n}\frac{\sum_{i=1}^{K}{w_ih_{r,i,n}^\ast}}{\left|\sum_{i=1}^{K}{w_ih_{r,i,n}^\ast}\right|}}\right\}\right]\nonumber\\
\!\!\overset{\mathrm{(c)}}{=}\!\Re\!\left\{\!N\!\cdot\!\mathbb{E}\!\left[\left|h_{p,n}^\ast\right|\right]\!\mathbb{E}\!\left[h_{r,m,n}\right]\!\mathbb{E}\!\!\left[\frac{\sum_{i=1}^{K}{w_ih_{r,i,n}^\ast}}{\left|\sum_{i=1}^{K}{w_ih_{r,i,n}^\ast}\right|}\right]\!\right\}\!=\!0,
\end{align}
where $\mathrm{(c)}$ exploits the independence of $\mathbf{h}_{r,m}$, $\mathbf{h}_{r,k}$, and $\mathbf{h}_{p}$, and the last equality comes from $h_{r,m,n}\!\sim\!\mathcal{CN}(0,1)$. Combined with the definition of $\ell_m$ in (\ref{interference}), we complete the proof.

\section{Proof of Theorem \ref{theorem2}}\label{appC}

Firstly, we express $\mathbb{E}\!\left[u_k^2 \right]$ in (\ref{th22}), where $\mathrm{(a)}$ exploits $\mathbb{E}\!\left[\left|h_{p,n}^\ast\right|^2\right]=1$ and the results in (\ref{Theorem2-E8-add1}). 
\begin{figure*}[t] 
\begin{align}
\mathbb{E}\left[u_k^2\right]&=\sum_{n=1}^{N}\mathbb{E}\left[\left|h_{p,n}^\ast\right|^2\right]\mathbb{E}\left[\left(\Re\left\{h_{r,k,n}\frac{\sum_{i=1}^{K}{w_ih_{r,i,n}^\ast}}{\left|\sum_{i=1}^{K}{w_ih_{r,i,n}^\ast}\right|}\right\}\right)^2\right]\nonumber\\
&\quad +\sum_{n=1}^{N}\sum_{n^\prime\neq n}\mathbb{E}\left[\left|h_{p,n}^\ast\right|\left|h_{p,n^\prime}^\ast\right|\right]\left(\mathbb{E}\left[\Re\left\{h_{r,k,n}\frac{\sum_{i=1}^{K}{w_ih_{r,i,n}^\ast}}{\left|\sum_{i=1}^{K}{w_ih_{r,i,n}^\ast}\right|}\right\}\right]\right)\left(\mathbb{E}\left[\Re\left\{h_{r,k,n^\prime}\frac{\sum_{i=1}^{K}{w_ih_{r,i,n^\prime}^\ast}}{\left|\sum_{i=1}^{K}{w_ih_{r,i,n^\prime}^\ast}\right|}\right\}\right]\right)\nonumber\\
&\overset{\mathrm{(a)}}{=}\sum_{n=1}^{N}\left\{\mathbb{E}\!\left[\left|h_{r,k,n}\frac{\sum_{i=1}^{K}w_i h_{r,i,n}^\ast}{\left|\sum_{i=1}^{K}w_i h_{r,i,n}^\ast\right|}\right|^2\right]\!-\mathbb{E}\!\left[\left(\!\Im\!\left\{h_{r,k,n}\frac{\sum_{i=1}^{K}w_ih_{r,i,n}^\ast}{\left|\sum_{i=1}^{K}w_ih_{r,i,n}^\ast\right|}\right\}\right)^2\right]\right\}+\frac{\pi N\!\left(N\!-\!1\right)}{4}\!\left(\!\frac{\sqrt\pi}{2}\frac{w_k}{\sqrt{\sum_{i=1}^{K}w_i^2}}\!\right)^2\nonumber \\
&=\sum_{n=1}^{N}\left(1-w_k^{-2}\mathbb{E}\!\left[\left(\Im\!\left\{a\frac{a^\ast+b^\ast}{\left|a^\ast+b^\ast\right|}\right\}\right)^2\right]\right)+\frac{\pi^2N\left(N-1\right)}{16}\frac{w_k^2}{\sum_{i=1}^{K}w_i^2},
\label{th22}
\end{align}
\hrulefill
\end{figure*}
Then, we obtain 
\begin{align}\label{Theorem3-E1-3}
\!\!\mathbb{E}\!\left[\!\left(\!\Im\!\left\{\!a\frac{a^\ast\!+\!b^\ast}{\left|a^\ast\!+\!b^\ast\right|}\!\right\}\!\right)^2\right]=\mathbb{E}\!\!\left[\!\left(\Im\left\{\frac{ab^\ast}{\left|a+b\right|}\right\}\right)^2\right]\nonumber\\
=\mathbb{E}\!\!\left[\frac{\left|a\right|^2\left|b\right|^2\sin^2\left(\angle a\!-\!\angle b\right)}{\left|a\right|^2\!+\!\left|b\right|^2\!+\!2\!\left|a\right|\!\left|b\right|\!\cos{\left(\angle a\!-\!\angle b\right)}}\!\right]~~~\nonumber\\
=\frac{1}{4}\mathbb{E}_{\left|a\right|,\left|b\right|}\!\left[q^2\mathbb{E}_z\!\!\left[\frac{\sin^2z}{p\!+\!q\cos{z}}\right]\right],~~~~~~~~~~
\end{align}
where $\angle a\!\sim\!{\mathcal{U}}\!\left(0,2\pi\right)$, $\angle b\!\sim\!{\mathcal{U}}\!\left(0,2\pi\right)$, $p\!=\!\left|a\right|^2\!+\!\left|b\right|^2$, $q\!=\!2\!\left|a\right|\!\left|b\right|$, and $z\!=\!\angle a\!-\!\angle b$. Utilizing the PDF in \emph{Lemma~\ref{lemma3}}, we have
\begin{align}\label{Theorem3-E1-4}
\mathbb{E}_z\!\!\left[\frac{\sin^2\!z}{p\!\!+\!\!q\!\cos{z}}\right]&\!\!=\!\!\int_{-2\pi}^{0}\!{\frac{\sin^2z}{p\!\!+\!\!q\!\cos{z}}\frac{2\pi\!\!+\!\!z}{4\pi^2}{\rm{d}}z}\!+\!\!\int_{0}^{2\pi}\!\!\!{\frac{\sin^2\!z}{p\!\!+\!\!q\!\cos{z}}\frac{2\pi\!\!-\!\!z}{4\pi^2}{\rm{d}}z}\nonumber\\
&\!\!\mathop=^{\left(\rm b\right)}\!\!\int_{0}^{2\pi}\!\!{\frac{\sin^2t}{p\!\!+\!\!q\!\cos{t}}\frac{t}{4\pi^2}{\rm{d}}t}\!+\!\!\int_{0}^{2\pi}\!\!{\frac{\sin^2z}{p\!\!+\!\!q\!\cos{z}}\frac{2\pi\!-\!z}{4\pi^2}{\rm{d}}z}\nonumber\\
% &\!\!=\!\!\frac{1}{2\pi}\!\int_{0}^{2\pi}\!{\frac{\sin^2z}{p+q\cos{z}}{\rm{d}}z}\nonumber\\
&\!\!=\!\!\frac{1}{2\pi}\!\int_{0}^{\pi}\!{\frac{\sin^2z}{p+q\cos{z}}{\rm{d}}z}\!+\!\frac{1}{2\pi}\!\int_{\pi}^{2\pi}\!{\frac{\sin^2z}{p+q\cos{z}}{\rm{d}}z}\nonumber\\
&\!\!\mathop=^{\left(\rm c\right)}\!\!\frac{1}{2\pi}\!\int_{0}^{\pi}\!{\frac{\sin^2z}{p+q\cos{z}}{\rm{d}}z}\!+\!\frac{1}{2\pi}\!\int_{0}^{\pi}\!{\frac{\sin^2\mu}{p-q\cos{\mu}}{\rm{d}}z}\nonumber\\
&\!\!\mathop=^{\left(\rm d\right)}\!\!\frac{1}{q^2}\left(p-\sqrt{p^2-q^2}\right),
\end{align}
where $({\rm b})$ is obtained by letting $t\!=\!2\pi\!+\!z$, $({\rm c})$ follows from $\mu\!=\!2\pi\!-\!z$, and $({\rm d})$ is calculated by using \cite[Eq.~(3.644.4)]{table}. Then, by substituting (\ref{Theorem3-E1-4}) into (\ref{Theorem3-E1-3}), we have
\begin{align}\label{Theorem3-E1-5}
\mathbb{E}\!\left[\!\left(\!\Im\!\left\{\!a\frac{a^\ast\!+\!b^\ast}{\left|a^\ast\!+\!b^\ast\right|}\!\right\}\!\right)^2\right]&\!=\!\frac{1}{4}\mathbb{E}_{\left|a\right|,\left|b\right|}\left[p-\sqrt{p^2-q^2}\right]\nonumber\\
&\!=\!\frac{1}{4}\mathbb{E}_{\left|a\right|,\left|b\right|}\!\left[\left(\left|a\right|^2\!+\!\left|b\right|^2\!-\!\left|\left|a\right|^2\!-\!\left|b\right|^2\right|\right)\right]\nonumber\\
&\!=\!\frac{1}{2}\mathbb{E}_{\left|a\right|,\left|b\right|}\!\left[\min\!\left\{\left|a\right|^2\!,\!\left|b\right|^2\right\}\right]\nonumber\\
&\!\mathop=^{\left(\rm e\right)}\!\frac{w_k^2\left(\sum_{i\neq k}w_i^2\right)}{2\sum_{i=1}^{K}w_i^2},
\end{align}
where $({\rm e})$ is obtained by using the fact that $\left|a\right|^2\!\sim\! {\rm Exp}\!\left(w_k^{-2}\right)$, $\left|b\right|^2\!\sim\! {\rm Exp}\!\left(\frac{1}{\sum_{i\neq k}w_i^2}\right)$ and applies \emph{Lemma~\ref{lemma4}}. Plugging (\ref{Theorem3-E1-5}) into (\ref{th22}), we have
\begin{align}\label{th26}
    \mathbb{E}[u_k^2]=\frac{N}{2}\!+\!\frac{8N+\pi^2N(N-1)}{16}\frac{w_k^2}{\sum_{i=1}^{K}w_i^2}.
\end{align}
Hence, we calculate the variance of $\ell_k$ as 
\begin{align}\label{eq_app_1}
\mathbb{V}\left [ \ell_k \right] & = \frac{\beta_k^2 p_k}{\lambda^2}\left(\mathbb{E}[u_k^2]-\left(\mathbb{E}[u_k]\right)^2\right)\nonumber\\
&=\frac{\beta_k^2 p_k}{\lambda^2}\left(\frac{N}{2}+\frac{(8-\pi^2)w_k^2N}{16\sum_{i=1}^{K}w_i^2}
\right).
\end{align}
Given that $\lambda$ scales with $N$, we conclude that $\mathbb{V}\left [ \ell_k \right]$ diminishes by the order of $\mathcal{O}\left(\frac{1}{N}\right)$.

For interference signals, we first calculate $\mathbb{E}[u_m^2]$ in (\ref{Theorem3-E3-1}),
\begin{figure*}[t] 
\begin{align}
\mathbb{E}\!\left[u_m^2\right]&=\mathbb{E}\!\left[\left(\Re\!\left\{\sum_{n=1}^{N}{\left|h_{p,n}^\ast\right|h_{r,m,n}\frac{\sum_{i=1}^{K}w_ih_{r,i,n}^\ast}{\left|\sum_{i=1}^{K}w_ih_{r,i,n}^\ast\right|}}\right\}\right)^2\right]\nonumber\\
&=\mathbb{E}\!\left[\left(\Re\!\left\{\sum_{n=1}^{N}{\left|h_{p,n}^\ast\right|\left|h_{r,m,n}\right|{\rm e}^{j\left(\delta_{1,n}+\delta_{2,n}\right)}}\right\}\right)^2\right]=\mathbb{E}\!\left[\left(\sum_{n=1}^{N}{\left|h_{p,n}^\ast\right|\left|h_{r,m,n}\right|\cos{\delta_n}}\right)^2\right]\nonumber\\
&=\mathbb{E}\!\left[\sum_{n=1}^{N}{\left|h_{p,n}^\ast\right|^2\left|h_{r,m,n}\right|^2\cos^2{\delta_n}}\right]\!+\!\mathbb{E}\!\left[\sum_{n=1}^{N}\sum_{n^\prime\neq n}{\left|h_{p,n}^\ast\right|\left|h_{r,m,n}\right|\left|h_{p,n^\prime}^\ast\right|\left|h_{r,m,n^\prime}\right|\cos{\delta_n}\cos{\delta_{n^\prime}}}\right]\nonumber\\
&=\sum_{n=1}^{N}\mathbb{E}\left[\cos^2{\delta_n}\right]+\frac{\pi^2}{16}\sum_{n=1}^{N}\sum_{n^\prime\neq n}\mathbb{E}\!\left[\cos{\delta_n}\right]\mathbb{E}\!\left[\cos{\delta_{n^\prime}}\right],
\label{Theorem3-E3-1}
\end{align}
\hrulefill
\end{figure*}
where $\delta_{1,n}\!=\!\angle h_{r,m,n}\!\sim\!{\mathcal{U}}\!\left(0,2\pi\right)$, $\delta_{2,n}\!=\!\angle\left(\sum_{i=1}^{K}w_ih_{r,i,n}^\ast\right)\!\sim\!{\mathcal{U}}\!\left(0,2\pi\right)$, $\delta_n\!=\!\delta_{1,n}\!+\!\delta_{2,n}$. Utilizing the PDF given in \emph{Lemma~\ref{lemma2}}, we have
\begin{align}\label{Theorem3-E3-2}
\mathbb{E}\!\left[\cos{\delta_n}\right]&\!=\!\int_{0}^{2\pi}{\frac{\delta_n}{4\pi^2}\cos\delta_n{\rm{d}}\delta_n}+\int_{2\pi}^{4\pi}{\frac{4\pi-\delta_n}{4\pi^2}\cos\delta_n{\rm{d}}\delta_n}\nonumber\\
&\!\mathop=^{\left(\rm f\right)}\!\int_{0}^{2\pi}{\frac{\delta_n}{4\pi^2}\cos \delta_n{\rm{d}}\delta_n}\!+\!\int_{0}^{2\pi}{\frac{4\pi\!-\!\left(t\!+\!2\pi\right)}{4\pi^2}\cos t{\rm{d}}t}\nonumber\\
&\!=\!\frac{1}{2\pi}\int_{0}^{2\pi}{\cos t{\rm{d}}t}=0,
\end{align}
where $({\rm f})$ is obtained by letting $t=-2\pi+\delta_n$. Similarly,
\begin{align}\label{Theorem3-E3-3}
\mathbb{E}\!\left[\cos^2{\delta_n}\right]&\!=\!\int_{0}^{2\pi}{\frac{\delta_n}{4\pi^2}\cos^2{\delta_n}{\rm{d}}\delta_n}\!+\!\int_{2\pi}^{4\pi}{\frac{4\pi\!-\!\delta_n}{4\pi^2}\cos^2{\delta_n}{\rm{d}}\delta_n}\nonumber\\
&\!=\!\frac{1}{2\pi}\int_{0}^{2\pi}{\cos^2 t{\rm{d}}t}=\frac{1}{2}.
\end{align}
Hence, by substituting (\ref{Theorem3-E3-2}) and (\ref{Theorem3-E3-3}) into  (\ref{Theorem3-E3-1}), we~have
\begin{align}\label{eq_app_2}
\mathbb{E}[u_m^2]=\frac{N}{2}.
\end{align}
Furthermore, the variance of $\ell_m$ is equal to
\begin{align}
    \mathbb{V}[\ell_m ]=\frac{\beta_m^2p_m N}{2\lambda^2}.
\end{align}
Similar to (\ref{eq_app_1}), we complete the proof.

\section{Proof of Theorem \ref{theorem3}}\label{appD}

Before deriving the MSE, we first calculate $\mathbb{E} [u_k u_{k^\prime}]$ in (\ref{th24}), where $k,k^\prime\!\in\!\mathcal{K}$, $k^\prime\!\neq\!k$, ${\rm A}\!\!=\!\!w_kh_{r,k,n}\frac{\sum_{i=1}^{K}\!{w_ih_{r,i,n}^\ast}}{\left|\sum_{i=1}^{K}\!{w_ih_{r,i,n}^\ast}\right|}$, and ${\rm B}\!\!=\!\!w_{k^\prime}h_{r,k^\prime\!,n}\frac{\sum_{i=1}^{K}\!{w_ih_{r,i,n}^\ast}}{\left|\sum_{i=1}^{K}\!{w_ih_{r,i,n}^\ast}\right|}$. 
\begin{figure*}[t]
    \begin{align}\label{th24}
    \mathbb{E}[u_k u_{k^\prime}]
&=\mathbb{E}\left[\sum_{n=1}^{N}{\left|h_{p,n}^\ast\right|^2\Re\left\{h_{r,k,n}\frac{\sum_{i=1}^{K}{w_ih_{r,i,n}^\ast}}{\left|\sum_{i=1}^{K}{w_ih_{r,i,n}^\ast}\right|}\right\}\Re\left\{h_{r,k^\prime,n}\frac{\sum_{i=1}^{K}{w_ih_{r,i,n}^\ast}}{\left|\sum_{i=1}^{K}{w_ih_{r,i,n}^\ast}\right|}\right\}}\right]\nonumber \\
&\quad+\mathbb{E}\left[\sum_{n=1}^{N}{\sum_{n^\prime\neq n}\left|h_{p,n}^\ast\right|\left|h_{p,n^\prime}^\ast\right|\Re\left\{h_{r,k,n}\frac{\sum_{i=1}^{K}{w_ih_{r,i,n}^\ast}}{\left|\sum_{i=1}^{K}{w_ih_{r,i,n}^\ast}\right|}\right\}\Re\left\{h_{r,k^\prime,n^\prime}\frac{\sum_{i=1}^{K}{w_ih_{r,i,n^\prime}^\ast}}{\left|\sum_{i=1}^{K}{w_ih_{r,i,n^\prime}^\ast}\right|}\right\}}\right]\nonumber \\
&=\sum_{n=1}^{N}\mathbb{E}\left[\Re\left\{h_{r,k,n}\frac{\sum_{i=1}^{K}{w_ih_{r,i,n}^\ast}}{\left|\sum_{i=1}^{K}{w_ih_{r,i,n}^\ast}\right|}\right\}\Re\left\{h_{r,k^\prime,n}\frac{\sum_{i=1}^{K}{w_ih_{r,i,n}^\ast}}{\left|\sum_{i=1}^{K}{w_ih_{r,i,n}^\ast}\right|}\right\}\right]+\frac{\pi^2N\left(N-1\right)}{16}\frac{w_kw_{k^\prime}}{\sum_{i=1}^{K}w_i^2}\nonumber \\
&=\sum_{n=1}^{N}w_k^{-1}w_{k^\prime}^{-1}\mathbb{E}\!\left[\Re\!\left\{{\rm A}\right\}\!\Re\!\left\{{\rm B}\right\}\right]\!+\!\frac{\pi^2N\!\left(\!N\!-\!1\right)}{16}\frac{w_kw_{k^\prime}}{\sum_{i=1}^{K}w_i^2},
    \end{align}
    \hrulefill
\end{figure*}
Then, we have
\begin{align}
&\!\!\mathbb{E}\!\left[\Re\!\left\{{\rm A}\right\}\!\Re\!\left\{{\rm B}\right\}\right]\!\mathop=^{\left(\rm a\right)}\!-\mathbb{E}\!\left[\Im\!\left\{{\rm A}\right\}\!\Im\!\left\{{\rm B}\right\}\right]\nonumber\\
&\!\!\!\!=\!-\!\frac{1}{2}\!\left\{\!\mathbb{E}\!\!\left[\left(\Im\!\left\{{\rm A}\!+\!{\rm B}\right\}\right)^2\right]\!-\!\mathbb{E}\!\!\left[\left(\Im\!\left\{{\rm A}\right\}\right)^2\right]\!-\!\mathbb{E}\!\!\left[\left(\Im\!\left\{{\rm B}\right\}\right)^2\right]\!\right\},
\label{Theorem4-E2-2}
\end{align}
where $({\rm a})$ comes from the fact that $\Re\!\left\{\mathbb{E}\left[{\rm A}^\ast{\rm B}\right]\right\}\!=\!0$. By letting $c\!\!=\!\!w_kh_k\!+\!w_{k^\prime}h_{k^\prime}\!\sim\!\mathcal{CN}\!\left(0,w_k^2\!+\!w_{k^\prime}^2\right)$ and $d\!=\!\sum_{i\neq k,k^\prime}{w_ih_i}\!\sim\!\mathcal{CN}\!\left(0,\sum_{i\neq k,k^\prime}w_i^2\right)$ and referring to (\ref{Theorem3-E1-5}), we first obtain
\begin{align}
&\mathbb{E}\!\left[\left(\Im\left\{{\rm A}\!+\!{\rm B}\right\}\right)^2\right]\!=\!\mathbb{E}\!\left[\left(\Im\left\{c\frac{c^\ast+d^\ast}{\left|c^\ast+d^\ast\right|}\right\}\right)^2\right]\nonumber\\
\!=\!&\frac{1}{2}\mathbb{E}_{\left|c\right|,\left|d\right|}\!\left[\min\left\{\left|c\right|^2,\left|d\right|^2\right\}\right]\!=\!\frac{\left(w_k^2\!+\!w_{k^\prime}^2\right)\sum_{i\neq k,k^\prime}w_i^2}{2\sum_{i=1}^{K}w_i^2}.
\label{Theorem4-E2-3}
\end{align}
Similarly, we have
\begin{align}
\mathbb{E}\!\left[\left(\Im\!\left\{{\rm A}\right\}\right)^2\right]=\frac{w_k^2\sum_{i\neq k}w_i^2}{2\sum_{i=1}^{K}w_i^2},
\label{Theorem4-E2-4}
\end{align}
and
\begin{align}
\mathbb{E}\!\left[\left(\Im\!\left\{{\rm B}\right\}\right)^2\right]=\frac{w_{k^\prime}^2\sum_{i\neq k^\prime}w_i^2}{2\sum_{i=1}^{K}w_i^2}.
\label{Theorem4-E2-5}
\end{align}
By combining all the above results, we obtain
\begin{align}\label{th25}
\mathbb{E}\left[u_ku_k^{\prime}\right]=\frac{N}{2}\frac{w_kw_{k^\prime}}{\sum_{i=1}^{K}w_i^2}+\frac{\pi^2N\left(N-1\right)}{16}\frac{w_kw_{k^\prime}}{\sum_{i=1}^{K}w_i^2}.
\end{align}

Next, for $\mathbb{E} [u_m u_{m^\prime}]$, $m,m^\prime\!\in\!\mathcal{M}$ and $m^\prime\!\neq\!m$,  we calculate it in (\ref{Theorem3-E4}). Similar to (\ref{Theorem3-E4}), we derive that $\mathbb{E}[u_k u_m]=0$.
\begin{figure*}[t] 
\begin{align}
\mathbb{E}\left[{u}_m{u}_{m^\prime}\right]\!&=\!\mathbb{E}\!\!\left[\!\Re\!\left\{\!\sum_{n=1}^{N}{\left|h_{p,n}^\ast\right|h_{r,m^\prime,n}\frac{\Re\!\left\{\mathbf{h}_p^H\mathbf{\Theta}\mathbf{h}_{r,m}\right\}\sum_{i=1}^{K}w_ih_{r,i,n}^\ast}{\left|\sum_{i=1}^{K}w_ih_{r,i,n}^\ast\right|}}\right\}\!\right]\!=\!\sum_{n=1}^{N}\Re\!\left\{\!\mathbb{E}\!\!\left[\left|h_{p,n}^\ast\right|h_{r,m^\prime,n}\frac{\Re\left\{\mathbf{h}_p^H\mathbf{\Theta}\mathbf{h}_{r,m}\right\}\sum_{i=1}^{K}w_ih_{r,i,n}^\ast}{\left|\sum_{i=1}^{K}w_ih_{r,i,n}^\ast\right|}\!\right]\!\right\}\nonumber\\
&=\sum_{n=1}^{N}\Re\left\{\mathbb{E}\left[h_{r,m^\prime,n}\right]\mathbb{E}\left[\left|h_{p,n}^\ast\right|\frac{\Re\left\{\mathbf{h}_p^H\mathbf{\Theta}\mathbf{h}_{r,m}\right\}\sum_{i=1}^{K}w_ih_{r,i,n}^\ast}{\left|\sum_{i=1}^{K}w_ih_{r,i,n}^\ast\right|}\right]\right\}=0.
\label{Theorem3-E4}
\end{align}
\hrulefill
% \vspace{-0.3cm}
\end{figure*}

\begin{figure*}[t] 
\begin{align}
{\rm MSE}=&\mathbb{E}\left[\left\Vert\sum_{k\in\mathcal{K}}{{\bar{h}}_k\mathbf{g}_{t,k}}+\sum_{m\in\mathcal{M}}{{\bar{h}}_m{\tilde{\mathbf{g}}}_{t,m}}+{\bar{\mathbf{z}}}_t\right\Vert^2\right]\nonumber\\
=&\sum_{k\in\mathcal{K}}\mathbb{E}\left[\left({\bar{h}}_k\right)^2\right]\mathbb{E}\left[\Vert\mathbf{g}_{t,k}\Vert^2\right]+\sum_{k\in\mathcal{K}}\sum_{k^\prime\neq k}\mathbb{E}\left[{\bar{h}}_k{\bar{h}}_{k^\prime}\right]\mathbb{E}\left[\mathbf{g}_{t,k}^T\mathbf{g}_{t,k^\prime}\right]+\sum_{m\in\mathcal{M}}\mathbb{E}\left[\left({\bar{h}}_m\right)^2\right]\mathbb{E}\left[\Vert{\tilde{\mathbf{g}}}_{t,m}\Vert^2\right]\nonumber\\
&+\sum_{m\in\mathcal{M}}\sum_{m^\prime\neq m}\mathbb{E}\left[{\bar{h}}_m{\bar{h}}_{m^\prime}\right]\mathbb{E}\left[{\tilde{\mathbf{g}}}_{t,m}^T{\tilde{\mathbf{g}}}_{t,m^\prime}\right]+2\sum_{k\in\mathcal{K}}\sum_{m\in\mathcal{M}}\mathbb{E}\left[{\bar{h}}_k{\bar{h}}_m\right]\mathbb{E}\left[\mathbf{g}_{t,k}^T{\tilde{\mathbf{g}}}_{t,m}\right]+\mathbb{E}\left[\left\Vert{\bar{\mathbf{z}}}_t\right\Vert^2\right],
\label{MSEnew}
\end{align}
\hrulefill
\end{figure*}

By substituting (\ref{gradient}) and (\ref{ghata}), we reformulate the MSE in~(\ref{MSE}) as (\ref{MSEnew}), where ${\bar{h}}_k\triangleq\ell_k-\frac{1}{K}$ and ${\bar{h}}_m\triangleq\ell_m$.
Based on the above results, we first calculate the MSE for Scheme \Rmnum{1} as follows.

\emph{1) Calculate $\mathbb{E}\left[\left({\bar{h}}_k\right)^2\right]$}: We have
\begin{align}\label{Theorem3-E1}
\mathbb{E}\!\left[\!\left({\bar{h}}_k\right)^2\!\right]&\!=\!\mathbb{E}\!\left[\!\left(\!\ell_k\!-\!\frac{1}{K}\!\right)^2\right]\mathop=^{\left(\rm b\right)}\mathbb{E}\!\left[\left(\ell_k\!\right)^2\right]\!-\!\frac{1}{K^2}\nonumber\\
&\!=\!\left(\!\frac{4}{\pi N\!\sqrt K}\!\right)^2\!\mathbb{E}\!\left[u_k^2\right]\!-\!\frac{1}{K^2}\overset{\mathrm{(c)}}{=}\!\frac{8\left(K\!+\!1\right)\!-\!\pi^2}{\pi^2NK^2},
\end{align}
where $({\rm b})$ comes from $\mathbb{E}\left[{\ell_k}\right]\!=\!\frac{1}{K}$, and $\mathrm{(c)}$ comes from (\ref{th26}).

\emph{2) Calculate $\mathbb{E}\left[{\bar{h}}_k{\bar{h}}_{k^\prime}\right]$}: Similar to (\ref{Theorem3-E1}), we have
\begin{align}\label{Theorem3-E2}
\mathbb{E}\!\left[{\bar{h}}_k{\bar{h}}_{k^\prime}\right]&\!=\!\left(\!\frac{4}{\pi N\!\sqrt K}\!\right)^2\!\!\mathbb{E}\!\left[u_k u_{k^\prime}\right]\!-\!\frac{1}{K^2}\overset{\mathrm{(d)}}{=}\frac{8-\pi^2}{\pi^2NK^2},
\end{align}
where $\mathrm{(d)}$ comes from (\ref{th25}).

\emph{3) Calculate $\mathbb{E}\left[\left({\bar{h}}_m\right)^2\right]$}: According to (\ref{eq_app_2}), we have
\begin{align}\label{Theorem3-E3-4}
\mathbb{E}\!\left[\!\left({\bar{h}}_m\right)^2\!\right]\!=\!\frac{\beta_m^2 P_m}{\lambda^2}\mathbb{E}\!\left[u_m^2\right]\!=\!\frac{8G^2\alpha_m^{2}}{\pi^2NK\min\limits_{k\in\mathcal{K}}{\alpha_k^2}}.
\end{align}

\emph{4) Calculate $\mathbb{E}\!\left[{\bar{h}}_m{\bar{h}}_{m^\prime}\right]$}: We have $\mathbb{E}\!\left[{\bar{h}}_m{\bar{h}}_{m^\prime}\right]\!=\!0$ from (\ref{Theorem3-E4}).

\emph{5) Calculate $\mathbb{E}\!\left[{\bar{h}}_k{\bar{h}}_m\right]$}: Similarly, we have $\mathbb{E}\!\left[{\bar{h}}_k{\bar{h}}_m\right]\!=\!0$.

\emph{6) Calculate $\mathbb{E}\left[\left\Vert{\bar{\mathbf{z}}}_t\right\Vert^2\right]$}: Finally, we have
\begin{align}
\mathbb{E}\!\left[\left\Vert{\bar{\mathbf{z}}}_t\right\Vert^2\right]=\frac{\sigma^2D}{2\lambda^2}=\frac{8G^2\sigma^2D}{\pi^2N^2K\min\limits_{k\in\mathcal{K}}{\alpha_k^2}}.
\label{Theorem3-E6}
\end{align}

Therefore, by substituting all the derived expectations into (\ref{MSEnew}), we obtain the MSE in (\ref{MSE1}).

Next, for Scheme \Rmnum{2}, we derive its MSE as follows, where similar results $\mathbb{E}\!\left[{\bar{h}}_m{\bar{h}}_{m^\prime}\right]\!=\!0$ and $\mathbb{E}\!\left[{\bar{h}}_k{\bar{h}}_m\right]\!=\!0$ are omitted.

\emph{1) Calculate $\mathbb{E}\left[\left({\bar{h}}_k\right)^2\right]$}: Similar to Scheme \Rmnum{1}, we have
\begin{align}\label{Theorem4-E1-3}
\mathbb{E}\!\left[\left({\bar{h}}_k\right)^2\right]&\!=\!\mathbb{E}\left[\left(\ell_k\right)^2\right]-\frac{1}{K^2}\!=\!\left(\frac{\alpha_k}{\lambda G}\right)^2\mathbb{E}\!\left[u_k^2\right]-\frac{1}{K^2}\nonumber \\
&=\frac{8\left(\alpha_k^2 \sum_{i=1}^{K}\alpha_i^{-2}+1\right)-\pi^2}{\pi^2NK^2}.
\end{align}

\emph{2) Calculate $\mathbb{E}\left[{\bar{h}}_k{\bar{h}}_{k^\prime}\right]$}: We have
\begin{align}\label{Theorem4-E2}
\mathbb{E}\!\left[{\bar{h}}_k{\bar{h}}_{k^\prime}\!\right]\!&\!=\!\frac{\alpha_k\alpha_{k^\prime}}{\left(\lambda G\right)^2}\mathbb{E}\!\left[u_ku_{k^\prime}\right]\!\!-\!\!\frac{1}{K^2}=\frac{8-\pi^2}{\pi^2NK^2}.
\end{align}

\emph{3) Calculate $\mathbb{E}\left[\left({\bar{h}}_m\right)^2\right]$}: We have
\begin{align}\label{Theorem4-E3}
\mathbb{E}\!\left[\left({\bar{h}}_m\right)^2\right]\!=\!\frac{N\beta_m^2 P_m}{2\lambda^2} \!=\!\frac{8G^2\alpha_m^2\sum_{i=1}^{K}\alpha_i^{-2}}{\pi^2NK^2}.
\end{align}

\emph{4) Calculate $\mathbb{E}\left[\left\Vert{\bar{\mathbf{z}}}_t\right\Vert^2\right]$}: Finally, we have
\begin{align}
\mathbb{E}\!\left[\left\Vert{\bar{\mathbf{z}}}_t\right\Vert^2\right]=\frac{\sigma^2D}{2\lambda^2}=\frac{8G^2\sigma^2D\sum_{i=1}^{K}\alpha_i^{-2}}{\pi^2N^2K^2}.
\label{Theorem4-E6}
\end{align}

Therefore, by substituting all the derived expectations into (\ref{MSEnew}), we complete the proof.

% \vspace{-8.pt}
\section{Proof of Theorem \ref{theorem4}}\label{appE}
\begin{figure*}
\begin{align} \label{eq69}
    \mathbb{E}\left[F(\mathbf{w}_{t+1})-F(\mathbf{w}_t)\right] 
    \leq -\left(\eta_t-\frac{L\eta_t^2}{2}\right) \mathbb{E}\left[\left\Vert \nabla F(\mathbf{w}_t) \right \Vert^2 \right] +\frac{L\eta_t^2}{2} \mathbb{E}\left[\left\Vert  \mathbf{g}_t-\nabla F(\mathbf{w}_t) \right \Vert^2\right]+ \frac{L\eta_t^2}{2} \mathrm{MSE}.
\end{align}
\hrulefill
% \vspace{-0.24cm}
\end{figure*}
Firstly, by exploiting \emph{Assumption 2}, we easily verify that $\mathbb{E}\!\left[ \mathbf{g}_t\right]\!\!=\!\!\nabla F(\mathbf{w}_t)$ and hence we conclude that the obtained $\hat{\mathbf{g}}_t$ is an unbiased estimation of the global gradient $\nabla F(\mathbf{w}_t)$. Building upon \emph{Assumption 1} and the unbiased gradient estimation, we perform the similar steps in \cite[Eq. (28)]{imperfectCSI} and obtain (\ref{eq69}).

For the second term in (\ref{eq69}), we exploit \emph{Assumption 2} and bound it by
\begin{align}\label{eq70}
     \mathbb{E}\!\left[\left\Vert  \mathbf{g}_t\!-\!\nabla F(\mathbf{w}_t) \right \Vert^2\right] &=  \frac{1}{K^2} \mathbb{E}\!\left[\left\Vert \sum_{k\in \mathcal{K}}\left(\mathbf{g}_{t,k}\!-\!\nabla F_k(\mathbf{w}_t) \right)\right \Vert^2\right]\nonumber \\
     &\leq \frac{1}{K^2}\sum_{k\in\mathcal{K}}\chi^2=\frac{\chi^2}{K}.
\end{align}

As for the MSE term, we first bound $\mathbb{E}\left[ \Vert \mathbf{g}_{t,k} \Vert^2 \right]$ by
\begin{align}
    &\mathbb{E}\!\left[ \Vert \mathbf{g}_{t,k} \Vert^2 \right]=\mathbb{E}\left[\Vert\mathbf{g}_{t,k}-\nabla F_k(\mathbf{w}_t)+\nabla F_k(\mathbf{w}_t)\Vert^2\right]\nonumber \\
    &\overset{\mathrm{(a)}}{\leq}\!  \mathbb{E}\!\left[\left\Vert \nabla F_k(\mathbf{w}_t) \right \Vert^2 \right]\!+\!\chi^2\!\overset{\mathrm{(b)}}{\leq} \!\xi^2  \mathbb{E}\!\left[\left\Vert \nabla F(\mathbf{w}_t) \right \Vert^2 \right]\!+\!\chi^2,
\end{align}
where $\mathrm{(a)}$ comes from \emph{Assumption 2} and $\mathrm{(b)}$ comes from \emph{Assumption 3}. Then, the cross term is bounded by
\begin{align}
    &\mathbb{E}\left[\mathbf{g}_{t,k}^T\mathbf{g}_{t,k^\prime}\right]=\nabla^T \!F_k(\mathbf{w}_t)\nabla F_{k^\prime}(\mathbf{w}_t)\nonumber \\
    &\overset{\mathrm{(c)}}{\leq} \Vert \nabla F_{k}(\mathbf{w}_t)\Vert \cdot \Vert \nabla F_{k^\prime}(\mathbf{w}_t)\Vert \overset{\mathrm{(d)}}{\leq} \xi^2  \mathbb{E}\left[\left\Vert \nabla F(\mathbf{w}_t) \right \Vert^2 \right],
\end{align}
where $\mathrm{(c)}$ is due to the Cauchy-Schwarz inequality and $\mathrm{(d)}$ also comes from \emph{Assumption 3}. Hence, we derive the upper bound of the MSE in Scheme \Rmnum{1} as 
\begin{align}\label{eq73}
    \mathrm{MSE}_1\!\leq & \frac{(16\!-\!\pi^2)\xi^2}{\pi^2N} \mathbb{E}\left[\left\Vert \nabla F(\mathbf{w}_t) \right \Vert^2 \right]\!+\!\frac{8(K\!+\!1)\!-\!\pi^2}{\pi^2 NK}\chi^2\nonumber \\
    &+\!\Delta_{1,2}\!+\!\Delta_{1,3}.
\end{align}
Similarly, the MSE of Scheme \Rmnum{2} is bounded by
\begin{align}\label{eq74}
    \mathrm{MSE}_2\!\leq &\frac{\frac{8}{K^2}\!\!\sum_{k\in\mathcal{K}}\alpha_k^{2}\sum_{i\in\mathcal{K}}\alpha_i^{-2}\!\!+\!8\!-\!\!\pi^2}{\pi^2 N}\xi^2  \mathbb{E}\!\left[\left\Vert \nabla F(\mathbf{w}_t) \right \Vert^2 \right]\nonumber \\
    &\!\!\!\!\!\!\!\!\!\!\!\!\!\!+\frac{\frac{8}{K}\!\!\sum_{k\in\mathcal{K}}\alpha_k^{2}\sum_{i\in\mathcal{K}}\alpha_i^{-2}\!+\!8\!-\!\pi^2}{\pi^2 NK}\chi^2\!+\!\Delta_{2,2}\!+\!\Delta_{2,3}.
\end{align}

Combining the results in (\ref{eq69}), (\ref{eq70}) and (\ref{eq73}), and setting $\eta_t =\frac{1}{\varpi_{1}\sqrt{T}}$, we evaluate the FL convergence with Scheme \Rmnum{1} as 
\begin{align}
    \frac{1}{T}\!\sum_{t=0}^{T-1}\!\mathbb{E}\!\!\left[\left\Vert \nabla F(\mathbf{w}_t) \right \Vert^2 \right] &\!\!\overset{\mathrm{(e)}}{\leq} \!\frac{2\varpi_{1}}{\sqrt{T}}\!\!\left(F(\mathbf{w}_0)\!-\!\mathbb{E}\!\left[ F(\mathbf{w}_{T})\right] \!+\!\frac{\varepsilon_1}{2\varpi_1^2} \right)\nonumber \\
    &\!\!\!\!\!\!\!\!\!\!\!\!\!\!\overset{\mathrm{(f)}}{\leq} \!\frac{2\varpi_{1}}{\sqrt{T}}\!\!\left(\!\!F(\mathbf{w}_0)\!-\!\mathbb{E}\!\left[ F(\mathbf{w}^*)\right] \!+\!\frac{\varepsilon_1}{2\varpi_1^2} \right),
\end{align}
where $\mathrm{(e)}$ is due to the fact that $\frac{1}{\sqrt{T}}\!-\!\frac{1}{2T}\!\!\geq\!\!\frac{1}{2\sqrt{T}}$ and  $\mathrm{(f)}$ is because $F(\mathbf{w}^*)\!\!\leq\!\!F(\mathbf{w}_T)$. Similarly, we obtain the convergence result for Scheme~\Rmnum{2} and complete the proof.

\end{appendices}

% \vspace{-8.pt}


\begin{thebibliography}{}
\providecommand{\url}[1]{#1}
\csname url@samestyle\endcsname
\providecommand{\newblock}{\relax}
\providecommand{\bibinfo}[2]{#2}
\providecommand{\BIBentrySTDinterwordspacing}{\spaceskip=0pt\relax}
\providecommand{\BIBentryALTinterwordstretchfactor}{4}
\providecommand{\BIBentryALTinterwordspacing}{\spaceskip=\fontdimen2\font plus
\BIBentryALTinterwordstretchfactor\fontdimen3\font minus
  \fontdimen4\font\relax}
\providecommand{\BIBforeignlanguage}[2]{{%
\expandafter\ifx\csname l@#1\endcsname\relax
\typeout{** WARNING: IEEEtran.bst: No hyphenation pattern has been}%
\typeout{** loaded for the language `#1'. Using the pattern for}%
\typeout{** the default language instead.}%
\else
\language=\csname l@#1\endcsname
\fi
#2}}
\providecommand{\BIBdecl}{\relax}
\BIBdecl

\end{thebibliography}


\begin{thebibliography}{1}
\bibliographystyle{IEEEtran}

\bibitem{2}
W. Xu \emph{et al.}, ``Edge learning for B5G networks with distributed signal processing: Semantic communication, edge computing, and wireless sensing," \emph{IEEE J. Sel. Topics Signal Process.}, vol. 17, no. 1, pp. 9--39, Jan. 2023.

\bibitem{3}
T. Li, A. K. Sahu, A. Talwalkar, and V. Smith, ``Federated learning: Challenges, methods, and future directions," \emph{IEEE Signal Process. Mag.}, vol. 37, no. 3, pp. 50--60, May 2020.

\bibitem{YJC}
J. Yao, W. Xu, Z. Yang, X. You, M. Bennis, and H. V. Poor, ``Wireless federated learning over resource-constrained networks: Digital versus analog transmissions," \emph{IEEE Trans. Wireless Commun.}, vol. 23, no. 10, pp. 14020--14036, Oct. 2024.

\bibitem{4}
N. Shlezinger, M. Chen, Y. C. Eldar, H. V. Poor, and S. Cui, ``UVeQFed: Universal vector quantization for federated learning," \emph{IEEE Trans. Signal Process.}, vol. 69, pp. 500--514, 2021.

\bibitem{5}
Z. Yang \emph{et al.}, ``Energy efficient federated learning over wireless communication networks," \emph{IEEE Trans. Wireless Commun.}, vol. 20, no. 3, pp. 1935--1949, Mar. 2021.

\bibitem{wireless1}
S. Zheng, C. Shen, and X. Chen, ``Design and analysis of uplink and downlink communications for federated learning," \emph{IEEE J. Sel. Areas Commun.}, vol. 39, no. 7, pp. 2150--2167, Jul. 2021.

\bibitem{wireless2}
M. M. Amiri and D. G\"{u}nd\"{u}z, ``Machine learning at the wireless edge: Distributed stochastic gradient descent over-the-air," \emph{IEEE Trans. Signal Process.}, vol. 68, pp. 2155--2169, 2020.

\bibitem{Eldar1}
T. Gafni, K. Cohen, and Y. C. Eldar, ``Federated learning from heterogeneous data via controlled Bayesian air aggregation," \emph{IEEE Trans. Signal Process.}, early access. doi: 10.1109/TSP.2024.3351469.

\bibitem{6}
J. Yao, Z. Yang, W. Xu, M. Chen, and D. Niyato, ``GoMORE: Global model reuse for rescource-constrained wireless federated learning," \emph{IEEE Wireless Commun. Lett.}, vol. 12, no. 9, pp. 1543--1547, Sept. 2023.

\bibitem{9}
K. Yang, T. Jiang, Y. Shi, and Z. Ding, ``Federated learning via over-the-air computation," \emph{IEEE Trans. Wireless Commun.}, vol. 19, no. 3, pp. 2022--2035, Mar. 2020.

\bibitem{900}
L. Qiao, Z. Gao, M. B. Mashhadi, and D. G\"{u}nd\"{u}z, ``Massive digital over-the-air computation for communication-efficient federated edge learning," \emph{IEEE J. Sel. Areas Commun.}, vol. 42, no. 11, pp. 3078--3094, Nov. 2024.

\bibitem{Eldar2}
T. Sery, N. Shlezinger, K. Cohen, and Y. C. Eldar, ``Over-the-air federated learning from heterogeneous data," \emph{IEEE Trans. Signal Process.}, vol. 69, pp. 3796--3811, Jun. 2021.

\bibitem{8}
G. Zhu, Y. Wang, and K. Huang, ``Broadband analog aggregation for low-latency federated edge learning," \emph{IEEE Trans. Wireless Commun.}, vol. 19, no. 1, pp. 491--506, Jan. 2020.

\bibitem{7}
T. Sery and K. Cohen, ``On analog gradient descent learning over multiple access fading channels," \emph{IEEE Trans. Signal Process.}, vol. 68, pp. 2897--2911, Apr. 2020.

\bibitem{pAirFL}
W. Shi, J. Yao, J. Xu, W. Xu, L. Xu, and C. Zhao, ``Empowering over-the-air personalized federated learning via RIS," \emph{Sci China Inf Sci}, vol. 67, no. 11, pp. 219302:1--2, Nov. 2024.

\bibitem{11}
X. Cao, G. Zhu, J. Xu, Z. Wang, and S. Cui, ``Optimized power control design for over-the-air federated edge learning," \emph{IEEE J. Sel. Areas Commun.}, vol. 40, no. 1, pp. 342--358, Jan. 2022. 

\bibitem{12}
A. Bereyhi \emph{et al.}, ``Device scheduling in over-the-air federated learning via matching pursuit," \emph{IEEE Trans. Signal Process.}, vol. 71, pp. 2188--2203, Jun. 2023.

\bibitem{15}
M. Kim, A. Lee Swindlehurst, and D. Park, ``Beamforming vector design and device selection in over-the-air federated learning," \emph{IEEE Trans. Wireless Commun.}, vol. 22, no. 11, pp. 7464--7477, Nov. 2023.

\bibitem{interference_JSAC}
Z. Wang, Y. Zhou, Y. Shi, and W. Zhuang, ``Interference management for over-the-air federated learning in multi-cell wireless networks," \emph{IEEE J. Sel. Areas Commun.}, vol. 40, no. 8, pp. 2361--2377, Aug. 2022.

\bibitem{multitaskFL}
C. Zhong, H. Yang, and X. Yuan, ``Over-the-air federated multi-task learning over MIMO multiple access channels," \emph{IEEE Trans. Wireless Commun.}, vol. 22, no. 6, pp. 3853--3868, Jun. 2023.

\bibitem{19}
Z. Wu, Q. Ling, T. Chen, and G. B. Giannakis, ``Federated variance-reduced stochastic gradient descent with robustness to Byzantine attacks," \emph{IEEE Trans. Signal Process.}, vol. 68, pp. 4583--4596, Jul. 2020.

\bibitem{20}
Z. Yang, A. Gang, and W. U. Bajwa, ``Adversary-resilient distributed and decentralized statistical inference and machine learning: An overview of recent advances under the Byzantine threat model," \emph{IEEE Signal Process. Mag.}, vol. 37, no. 3, pp. 146--159, May 2020.

\bibitem{22}
X. Fan, Y. Wang, Y. Huo, and Z. Tian, ``BEV-SGD: Best effort voting SGD against Byzantine attacks for analog-aggregation-based federated learning over the air," \emph{IEEE Internet Things J.}, vol. 9, no. 19, pp. 18946--18959, Oct. 2022.

\bibitem{24}
S. Park and W. Choi, ``Byzantine fault tolerant distributed stochastic gradient descent based on over-the-air computation," \emph{IEEE Trans. Commun.}, vol. 70, no. 5, pp. 3204--3219, May 2022.

\bibitem{25}
H. Sifaou and G. Y. Li, ``Robust federated learning via over-the-air computation," in \emph{Proc. IEEE 32nd Int. Workshop Mach. Learn. Signal Process. (MLSP)}, Xi'an, China, 2022, pp. 1--6.

\bibitem{27}
M. Di Renzo \emph{et al.}, ``Smart radio environments empowered by reconfigurable intelligent surfaces: How it works, state of research, and the road ahead," \emph{IEEE J. Sel. Areas Commun.}, vol. 38, no. 11, pp. 2450--2525, Nov. 2020.

\bibitem{28}
W. Shi, W. Xu, X. You, C. Zhao, and K. Wei, ``Intelligent reflection enabling technologies for integrated and green Internet-of-Everything beyond 5G: Communication, sensing, and security," \emph{IEEE Wireless Commun.}, vol. 30, no. 2, pp. 147--154, Apr. 2023.

\bibitem{29}
C. Pan \emph{et al.}, ``An overview of signal processing techniques for RIS/IRS-aided wireless systems," \emph{IEEE J. Sel. Topics Signal Process.}, vol. 16, no. 5, pp. 883--917, Aug. 2022.

\bibitem{jxu}
J. Xu \emph{et al.}, ``Reconfiguring wireless environment via intelligent surfaces for 6G: Reflection, modulation, and security," \emph{Sci China Inf Sci}, vol. 66, no. 3, pp. 130304:1--20, Mar. 2023.

\bibitem{30}
W. Shi, J. Xu, W. Xu, C. Yuen, A. L. Swindlehurst, and C. Zhao, ``On secrecy performance of RIS-assisted MISO systems over Rician channels with spatially random eavesdroppers," \emph{IEEE Trans. Wireless Commun.}, vol. 23, no. 8, pp. 8357--8371, Aug. 2024.

\bibitem{zyhe1}
Z. He, J. Xu, H. Shen, W. Xu, C. Yuen, and M. Di Renzo, ``Joint training and reflection pattern optimization for non-ideal RIS-aided multiuser systems," \emph{IEEE Trans. Commun.}, vol. 72, no. 9, pp. 5735--5751, Sept. 2024.

\bibitem{RIS-FL-1}
K. Yang, Y. Shi, Y. Zhou, Z. Yang, L. Fu, and W. Chen, ``Federated machine learning for intelligent IoT via reconfigurable intelligent surface," \emph{IEEE Netw.}, vol. 34, no. 5, pp. 16--22, Sept./Oct. 2020.

\bibitem{RIS-FL-2}
H. Liu, X. Yuan, and Y. -J. A. Zhang, ``CSIT-free model aggregation for federated edge learning via reconfigurable intelligent surface," \emph{IEEE Wireless Commun. Lett.}, vol. 10, no. 11, pp. 2440--2444, Nov. 2021.

\bibitem{RIS-FL-3}
B. Yang \emph{et al.}, ``Federated spectrum learning for reconfigurable intelligent surfaces-aided wireless edge networks," \emph{IEEE Trans. Wireless Commun.}, vol. 21, no. 11, pp. 9610--9626, Nov. 2022.

\bibitem{RIS-FL-4}
H. Li, R. Wang, W. Zhang, and J. Wu, ``One bit aggregation for federated edge learning with reconfigurable intelligent surface: Analysis and optimization," \emph{IEEE Trans. Wireless Commun.}, vol. 22, no. 2, pp. 872--888, Feb. 2023.

\bibitem{RIS-AirFL-1}
W. Ni, Y. Liu, Z. Yang, H. Tian, and X. Shen, ``Federated learning in multi-RIS-aided systems," \emph{IEEE Internet Things J.}, vol. 9, no. 12, pp. 9608--9624, Jun. 2022.

\bibitem{RIS-AirFL-2}
H. Liu, X. Yuan, and Y.-J. A. Zhang, ``Reconfigurable intelligent surface enabled federated learning: A unified communication-learning design approach," \emph{IEEE Trans. Wireless Commun.}, vol. 20, no. 11, pp. 7595--7609, Nov. 2021.

\bibitem{RIS-AirFL-3}
Z. Wang \emph{et al.}, ``Federated learning via intelligent reflecting surface," \emph{IEEE Trans. Wireless Commun.}, vol. 21, no. 2, pp. 808--822, Feb. 2022.

\bibitem{RIS-AirFL-4}
Z. Wang, Y. Zhou, Y. Zou, Q. An, Y. Shi, and M. Bennis, ``A graph neural network learning approach to optimize RIS-assisted federated learning," \emph{IEEE Trans. Wireless Commun.}, vol. 22, no. 9, pp. 6092--6106, Sept. 2023.

\bibitem{RIS-AirFL-5}
J. Zheng, H. Tian, W. Ni, W. Ni, and P. Zhang, ``Balancing accuracy and integrity for reconfigurable intelligent surface-aided over-the-air federated learning," \emph{IEEE Trans. Wireless Commun.}, vol. 21, no. 12, pp. 10964--10980, Dec. 2022.

\bibitem{RIS-AirFL-6}
J. Mao and A. Yener, ``ROAR-Fed: RIS-assisted over-the-air adaptive resource allocation for federated learning," in \emph{Proc. IEEE Int. Conf. Commun. (ICC)}, Rome, Italy, 2023, pp. 4341--4346.

\bibitem{FedNova}
J. Wang, Q. Liu, H. Liang, G. Joshi, and H. V. Poor, ``Tackling the objective inconsistency problem in heterogeneous federated optimization," in \emph{Proc. Adv. Neural Inf. Process. Syst. (NeurIPS)}, 2020, pp. 1--13.

\bibitem{link}
K. B. Letaief, Y. Shi, J. Lu, and J. Lu, ``Edge artificial intelligence for 6G: Vision, enabling technologies, and applications," \emph{IEEE J. Sel. Areas Commun.}, vol. 40, no. 1, pp. 5--36, Jan. 2022.

\bibitem{link1}
J. Zhang, H. Du, Q. Sun, B. Ai, and D. W. K. Ng, ``Physical layer security enhancement with reconfigurable intelligent surface-aided networks," \emph{IEEE Trans. Inf. Forensics Security}, vol. 16, pp. 3480--3495, May 2021.

\bibitem{link2}
Z. Xiao \emph{et al.}, ``Antenna array enabled space/air/ground communications and networking for 6G," \emph{IEEE J. Sel. Areas Commun.}, vol. 40, no. 10, pp. 2773--2804, Oct. 2022.

\bibitem{biased-ICML}
A. Ajalloeian and S. U. Stich, ``On the convergence of SGD with biased gradients," \emph{in Proc. 37th Int. Conf. Mach. Learn. (ICML)}, 2020.

\bibitem{35}
J. Ren \emph{et al.}, ``Scheduling for cellular federated edge learning with importance and channel awareness," \emph{IEEE Trans. Wireless Commun.}, vol. 19, no. 11, pp. 7690--7703, Nov. 2020.

\bibitem{Eldar3}
Z. Wang, K. Huang, and Y. C. Eldar, ``Spectrum breathing: Protecting over-the-air federated learning against interference," \emph{IEEE Trans. Wireless Commun.}, early access. doi: 10.1109/TWC.2024.3368197.

\bibitem{favorable}
Z. Chen and E. Bj\"{o}rnson, ``Channel hardening and favorable propagation in cell-free massive MIMO with stochastic geometry," \emph{IEEE Trans. Commun.}, vol. 66, no. 11, pp. 5205--5219, Nov. 2018.

\bibitem{36}
X. Wei, C. Shen, J. Yang, and H. V. Poor, ``Random orthogonalization for federated learning in massive MIMO systems," \emph{IEEE Trans. Wireless Commun.}, vol. 23, no. 3, pp. 2469--2485, Mar. 2024.

\bibitem{zyhe2}
Z. He, W. Xu, H. Shen, D. W. K. Ng, Y. C. Eldar, and X. You, ``Full-duplex communication for ISAC: Joint beamforming and power optimization," \emph{IEEE J. Sel. Areas Commun.}, vol. 41, no. 9, pp. 2920--2936, Sept. 2023.

\bibitem{adapt}
H. Wu and P. Wang, ``Fast-convergent federated learning with adaptive weighting," \emph{IEEE Trans. Cogn. Commun. Netw.}, vol. 7, no. 4, pp. 1078--1088, Dec. 2021.

\bibitem{RR}
H. H. Yang, Z. Liu, T. Q. S. Quek, and H. V. Poor, ``Scheduling policies for federated learning in wireless networks," \emph{IEEE Trans. Commun.}, vol. 68, no. 1, pp. 317--333, Jan. 2020.

\bibitem{minMSE}
F. Zhu, Y. Zhao, W. Xu, and X. You, ``CSIT-free model aggregation for multi-RIS-assisted over-the-air computation," in \emph{Proc. Int. Symp. Wireless Commun. Syst. (ISWCS)}, Hangzhou, China, 2022, pp. 1--5.

\bibitem{wshi}
W. Shi, J. Xu, W. Xu, M. Di Renzo, and C. Zhao, ``Secure outage analysis of RIS-assisted communications with discrete phase control," \emph{IEEE Trans. Veh. Technol.}, vol. 72, no. 4, pp. 5435--5440, Apr. 2023.

\bibitem{FedProx}
T. Li, A. K. Sahu, M. Zaheer, M. Sanjabi, A. Taiwalkar, and V. Smith, ``Federated optimization in heterogeneous networks," in \emph{Proc. Mach. Learn. Syst. (MLSys)}, 2020, pp. 429--450.

\bibitem{aRIS}
K. Zhi, C. Pan, H. Ren, K. K. Chai, and M. Elkashlan, ``Active RIS versus passive RIS: Which is superior with the same power budget?," \emph{IEEE Commun. Lett.}, vol. 26, no. 5, pp. 1150--1154, May 2022.

\bibitem{bivariate}
C. C. Tan and N. C. Beaulieu, ``Infinite series representations of the bivariate Rayleigh and Nakagami-m distributions," \emph{IEEE Trans. Commun.}, vol. 45, no. 10, pp. 1159--1161, Oct. 1997.

\bibitem{table}
I. S. Gradshteyn and I. M. Ryzhik, \emph{Table of Integrals, Series, and Products}, 7th ed. San Diego, CA, USA: Academic, 2007.

\bibitem{imperfectCSI}
J. Yao, Z. Yang, W. Xu, D. Niyato, and X. You, ``Imperfect CSI: A key factor of uncertainty to over-the-air federated learning," \emph{IEEE Wireless Commun. Lett.},  vol. 12, no. 12, pp. 2273--2277, Dec. 2023.

\end{thebibliography}
\end{document}